%% file: IPDPS22_main.tex
\documentclass[conference,10pt]{IEEEtran}
\IEEEoverridecommandlockouts
\usepackage{cite}
\usepackage{amsmath,amssymb,amsfonts}
\usepackage{graphicx}
\usepackage{textcomp}
\usepackage{xcolor}
\usepackage[font=small]{caption}
\usepackage{float}
\usepackage[font=small]{subcaption}
\usepackage{algorithm}
\usepackage[noend]{algpseudocode}
\usepackage{amsthm}
\newtheorem{theorem}{Theorem}
\theoremstyle{remark}
\newtheorem{remark}{Remark}
\newcommand{\norm}[1]{\left\lVert#1\right\rVert}
\DeclareMathAlphabet{\mathbfsl}{OT1}{ppl}{b}{it}

\usepackage{comment}

\def\name{AVCC~}
\def\BibTeX{{\rm B\kern-.05em{\sc i\kern-.025em b}\kern-.08em
    T\kern-.1667em\lower.7ex\hbox{E}\kern-.125emX}}

\makeatletter
\newcommand{\linebreakand}{%
  \end{@IEEEauthorhalign}
  \hfill\mbox{}\par
  \mbox{}\hfill\begin{@IEEEauthorhalign}
}
\makeatother

\begin{document}

\title{Adaptive Verifiable Coded Computing: Towards Fast, Secure and Private Distributed Machine Learning}
\begin{comment}
\author{}
\end{comment}
\author{\IEEEauthorblockN{Tingting Tang}
\IEEEauthorblockA{\textit{Computer Science Department} \\
\textit{University of Southern California}\\
Los Angeles, USA \\
tangting@usc.edu}
\and
\IEEEauthorblockN{Ramy E. Ali}
\IEEEauthorblockA{\textit{Electrical Engineering Department} \\
\textit{University of Southern California}\\
Los Angeles, USA \\
reali@usc.edu}
\and
\IEEEauthorblockN{Hanieh Hashemi}
\IEEEauthorblockA{\textit{Electrical Engineering Department} \\
\textit{University of Southern California}\\
Los Angeles, USA \\
hashemis@usc.edu}
\linebreakand
\IEEEauthorblockN{Tynan Gangwani}
\IEEEauthorblockA{\textit{Electrical Engineering Department} \\
\textit{University of Southern California}\\
Los Angeles, USA \\
tgangwan@usc.edu}
\and
\IEEEauthorblockN{Salman Avestimehr}
\IEEEauthorblockA{\textit{Electrical Engineering Department} \\
\textit{University of Southern California}\\
Los Angeles, USA \\
avestime@usc.edu}
\and
\IEEEauthorblockN{Murali Annavaram}
\IEEEauthorblockA{\textit{Electrical Engineering Department} \\
\textit{University of Southern California}\\
Los Angeles, USA \\
annavara@usc.edu}
}

\maketitle
\thispagestyle{plain}
\pagestyle{plain}
%\begin{abstract}
%This document is a model and instructions %for \LaTeX.
%This and the IEEEtran.cls file define the %components of your paper [title, text, %heads, etc.]. *CRITICAL: Do Not Use %Symbols, Special Characters, Footnotes, 
%or Math in Paper Title or Abstract.
%\end{abstract}
\input{abstract}


\input{Introduction}

\input{backgorund}
\input{SystemStructure}
\input{implementation}
\input{Experiments}
\input{Conclusion}
\section*{Acknowledgment}

This material is based upon work supported by Defense Advanced Research Projects Agency (DARPA) under Contract No. HR001117C0053 and FA8750-19-2-1005, ARO award W911NF1810400, NSF grants CCF-1703575, CCF-1763673, and MLWINS-2002874, ONR Award No. N00014-16-1-2189, and a gift from Intel/Avast/Borsetta via the PrivateAI institute, a gift from Cisco, and a gift from Qualcomm. The views, opinions, and/or findings expressed are those of the author(s) and should not be interpreted as representing the official views or policies of the Department of Defense or the U.S. Government.
\bibliographystyle{./IEEEtran}
\bibliography{./IPDPS22_main.bib}
\end{document}

%% file: abstract.tex
\begin{abstract}
Stragglers, Byzantine workers, and data privacy are the main bottlenecks in distributed cloud computing. Some prior works proposed coded computing strategies to jointly address all three challenges.  They require either a large number of workers, a significant communication cost or a significant computational complexity to tolerate Byzantine workers. Much of the overhead in prior schemes comes from the fact that they tightly couple coding for all three problems into a single framework.  In this paper, we propose \textbf{A}daptive \textbf{V}erifiable \textbf{C}oded \textbf{C}omputing (\textbf{AVCC}) framework that decouples the Byzantine node detection challenge from the straggler tolerance. AVCC leverages coded computing just for handling stragglers and privacy, and then uses an orthogonal approach that leverages verifiable computing to mitigate Byzantine workers. Furthermore, AVCC dynamically adapts its coding scheme to trade-off straggler tolerance with Byzantine protection.  We evaluate AVCC on a compute-intensive distributed logistic regression application. Our experiments show that AVCC achieves up to $4.2\times$ speedup and up to $5.1\%$ accuracy improvement over the state-of-the-art Lagrange coded computing approach (LCC). AVCC also speeds up the conventional uncoded implementation of distributed logistic regression by up to $7.6\times$, and improves the test accuracy by up to $12.1\%$. 
\end{abstract}

%% file: Introduction.tex
\section{Introduction}
Distributed machine learning using cloud resources is widely used as it allows users to offload their compute-intensive  operations to run on multiple cloud servers~\cite{xu2012cloud}. Distributed computing, however, faces several challenges such as stragglers and compromised systems. Execution speed variations are commonly observed among compute nodes in the cloud, which can be up to an order of magnitude resulting in straggler behavior. These variations are due to the heterogeneity in server hardware, resource contention across shared virtual instances, IO delays, or even hardware faults~\cite{ananthanarayanan2010reining}. Stragglers hamper the end-to-end system performance~\cite{dean2013tail}. The second challenge is that hackers routinely compromise some machines in the cloud. These compromised nodes cause two problems. The first problem is that users' data privacy may be compromised when some hacked cloud instances collude to extract private information. In the other words, hacked workers may collude to glean information about the data that a user does not want to disclose. Second, hacked nodes may act as Byzantine nodes and return incorrect computational results to the client that may derail the training performance~\cite{blanchard2017machine}.  The goal of this work is to provide a unified efficient framework that jointly tackles stragglers, provides data privacy and eliminates Byzantine nodes. 

Most prior works rely on replication to provide straggler resiliency ~\cite{Apache-Hadoop,suresh2015c3,shah2015redundant,ananthanarayanan2013effective,dean2013tail,gardner2015reducing}. Replication, however, entails significant overhead. Since it is unknown which node may be a straggler a priori, replication strategies may pro-actively assign the same task to multiple nodes. Alternatively, reactive strategies may wait for a straggler delay to appear and then relaunch the straggling task on another node, which delays the overall execution. Coded computing based approaches are known to be more efficient when stragglers are not known a priori~\cite{lee2017speeding, yu2019lagrange}. In such approaches, a primary server encodes the data and distributes the encoded data over the workers. The workers then perform the computations over the encoded data and the desired computation can be recovered from the fastest subset of workers. For instance, the coding-theoretic approach of \cite{lee2017speeding} uses a maximum distance separable (MDS) $(N, K)$-code for encoding the data. With $(N, K)$ MDS coding, the data is split into $K$ pieces and then encoded into $N$ pieces and distributed to $N$ workers to perform linear operations, such as matrix-vector multiplication. If a subset of $K$ nodes ($K \leq N$) returns the result to the primary server, it can decode the full result.  

More advanced encoding strategies mask the data with random noise with the joint aim of mitigating stragglers, ensuring data privacy as well as tackling Byzantine nodes. Specifically, Lagrange coded computing (LCC)~\cite{yu2019lagrange} provides straggler resiliency, Byzantine robustness and privacy protection even if a subset of workers, up to a certain size, collude. LCC guarantees that the colluding workers cannot learn any information about that data in the information-theoretic sense. However, the cost of tolerating Byzantine workers with LCC is twice as the cost of tolerating  stragglers. For instance, 
% in a distributed support vector machine (SVM) training,  
tolerating two Byzantine workers requires an additional four workers while tolerating two stragglers only requires two additional workers. As we describe in more detail later, recent works reduced the cost of tolerating Byzantine workers to be the same as the cost of tolerating stragglers at the expense of increasing the communication cost significantly \cite{yang2021coded} or a significant computation complexity  \cite{subramaniam2019collaborative,soleymani2021list}. 

Inspired by the prior coding based approaches, and motivated by the large overheads faced by these approaches, we propose the \emph{Adaptive Verifiable Coded Computing} (AVCC) framework that jointly addresses stragglers, Byzantine workers and data privacy. 
Unlike LCC, the cost of tolerating Byzantine workers in AVCC is the same as the cost of tolerating stragglers. \name achieves this improvement through a unique decoupling of the data encoding for tackling stragglers and privacy, and an orthogonal information-theoretic verifiable computing approach  that uses Freivalds' algorithm \cite{Freivalds1977ProbabilisticMC} to detect Byzantine workers.  This decoupling enables \name to tolerate stragglers and tackle untrusted nodes in any distributed polynomial computations. \name  further adapts to the dynamics of the system  by changing the coding strategy at runtime depending on the straggler or Byzantine prevalence. 

The basic intuition behind \name can be provided with the following example. Consider the case when \name uses $(N, K)$-MDS coding to tolerate stragglers. MDS coding can be considered as a simplified version of LCC that can be used for linear computations. The primary server encodes the data and sends it to $N$ workers. It then receives and decodes the fastest $K$ out of the $N$ worker results to compute the final result. However, \name's verification process checks the integrity of the computation provided by each of the $K$ workers (Byzantine Workers). If any one of the $K$ workers fails the verification process, AVCC tags such a worker as a Byzantine node and discards the results provided from that node. It then has to wait for additional workers whose results can be verified before decoding the full computational output. Thus, \name trades-off straggler tolerance for Byzantine detection
%As long as the number of Byzantine nodes is at most the same as straggler nodes VCC 
and correctly computes the result.
%, albeit with reduced performance.
%as it waits for more verified workers. 
%In particular, in this paper we apply VCC for  computations to machine learning applications such as linear regression and logistic regression and other applications  that use  matrix-matrix multiplications. VCC's     
Our experiments show that \name speeds up the state-of-the-art LCC implementation of distributed logistic regression by up to $4.2\times$, and improves the test accuracy by up to $5.1\%$ accuracy. AVCC also achieves up to $7.6\times$ speedup and up to $12.1\%$ accuracy improvement over the conventional uncoded approach of distributed logistic regression.

%The key idea of VCC is that it uses the Frievald's algorithm to verify the computation result of each worker. This ensures that VCC can detect Byzantine workers with high probability and with low complexity. This verification strategy allows the primary to verify the outcome of the fastest workers, while the other workers are working on their computations. In the coded computing approaches, however, the primary server has to wait for all workers except the stragglers to finish their computations before starting to handle Byzantine workers and recover the computation outcome. After detecting Byzantine workers, VCC removes them from the distributed system for the next round of the computation. Since removing these Byzantine workers requires a different encoding scheme designed for fewer workers, we also introduce a low overhead scheme to enable this dynamic reconfiguration. 
% To address the Byzantine workers in a distributed system different solutions are proposed~\cite{rajput2019detox, chen2018draco}.

\textbf{Organization.} The rest of this paper is organized as follows. Section~\ref{sec:background} provides a background about coded computing, discusses the closely-related works and our contributions. In Section~\ref{sec:model}, we describe our system, the threat model, and our guarantees. Section~\ref{sec:implementation} introduces our adaptive verifiable coded computing framework. In Section~\ref{sec:experimentalsetup}, we describe our experimental setup followed by extensive experiments to evaluate our method in Section~\ref{sec:experiments}. Finally, concluding remarks and future directions are discussed in Section~\ref{sec:conclusion}.

%% file: backgorund.tex
\section{Background and Related Works}
\label{sec:background}
In this section, we provide a brief background about coded computing, verifiable computing and  the closely-related works. 
\subsection{Coded Computing}
\label{Subsection:coded}
%Coded computing is a concept that allows computations to be performed on coded data. 
In the past, coding was used mainly to tolerate data losses during communication and data storage. However, coded computing has extended this concept to enable tolerating stragglers while performing computations. In a system with dataset $\mathbf X \in \mathbb F_q^{m \times d}$, where $m$ is the number of samples and $d$ is the feature size, the goal of distributed computing is to compute a multivariate polynomial $f: \mathbb V \to \mathbb U$ over $\mathbf X=(\mathbf X_1^{\top }, \mathbf X_2^{\top}, \cdots, \mathbf X_K^{\top})^{\top}$, where $\mathbf X_i \in \mathbb F_q^{m/K \times d}$ and $\mathbb V$ and $\mathbb U$ are vector spaces of dimensions $v$ and $u$, respectively, over a finite field $\mathbb F_q$. That is, the goal is to compute $f(\mathbf X_i), \forall i \in [K]$ using a distributed collection of compute nodes. We describe two approaches for encoding data: MDS which is used for linear computations, and LCC which can handle any polynomial computations. 

\textbf{MDS Coding.} MDS coded computing is a computing paradigm that enables distributed computing on encoded data to tolerate stragglers. Fig. \ref{fig:MDS} illustrates the idea of computing a matrix-vector multiplication $\mathbf X \mathbfsl b$ using $3$ workers with a $(3,2)$ MDS code. The data matrix $\mathbf X$ is evenly divided into $2$ sub-matrices $\mathbf X_1$ and $\mathbf X_2$, then encoded into $3$ coded matrices $\tilde{\mathbf X}_1=\mathbf X_1, \tilde{\mathbf X}_2=\mathbf X_2$ and $\tilde{\mathbf X}_3=\mathbf X_1 + \mathbf X_2$, and assigned to worker $1, 2$, and $3$, respectively. Worker $i$ then receives the coded matrix $\tilde{\mathbf X}_i$ and the vector $\mathbfsl b$ and starts computing $\tilde{\mathbf X}_i \mathbfsl b$, where $i \in [3]$.  The final result can be recovered when the results from any $2$ out of the $3$ workers are received, without the need to wait for one straggler. Assume that results from worker $2$ and worker $3$ are received. Then $\mathbf X_1\mathbfsl b$ can be decoded by subtracting $\mathbf X_2 \mathbfsl b$ from $(\mathbf X_1+\mathbf X_2) \mathbfsl b$, and the final result can be obtained by concatenating $\mathbf X_1\mathbfsl b$ and $\mathbf X_2\mathbfsl b$.
\begin{figure}[t]
    \centering
    \includegraphics[width=0.4\textwidth]{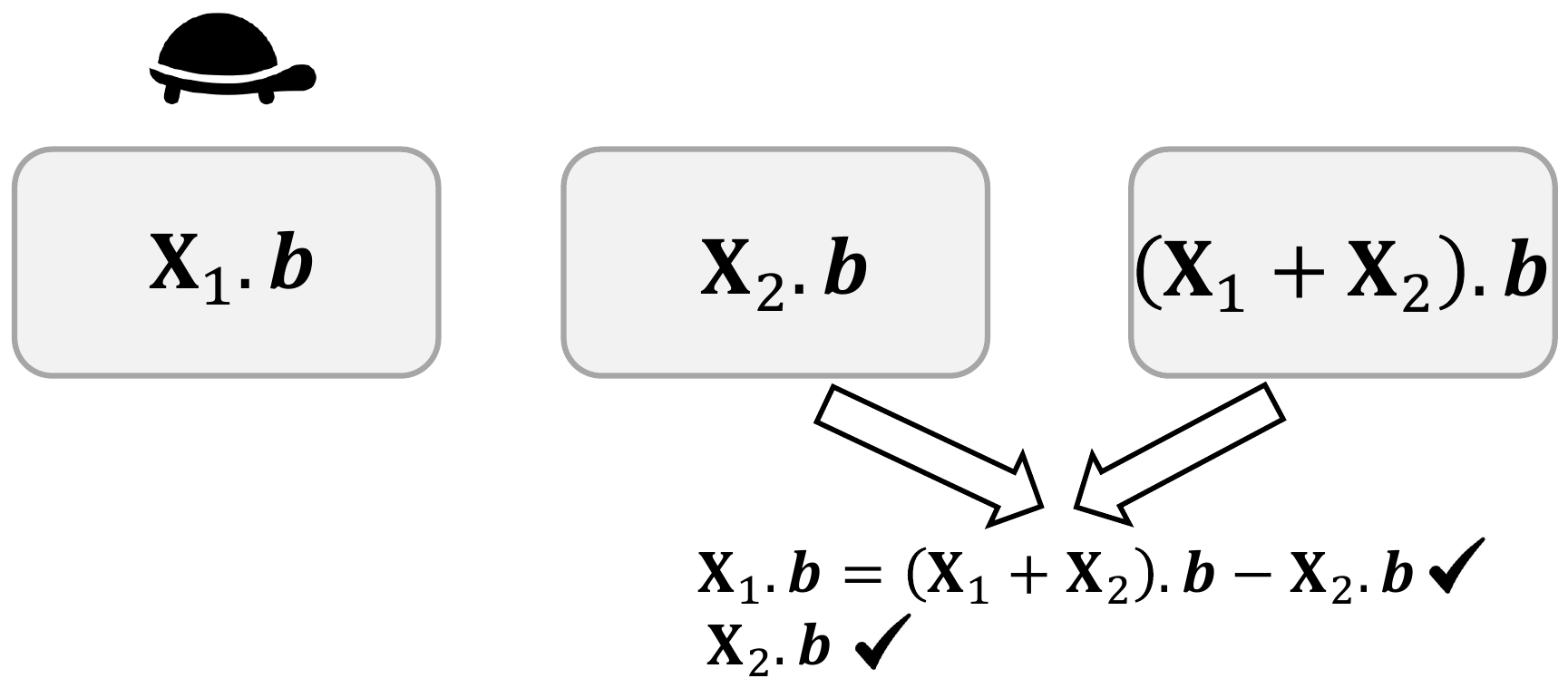}
    \caption{An illustration of a distributed computing system using $(3,2)$ MDS code is depicted. The goal is to compute the matrix-vector multiplication $\mathbf X \mathbfsl b$, where $\mathbf X=[\mathbf X_1^\top, \mathbf X_2^\top]^\top$ while tolerating one straggler. In this example, the first worker is a straggler and only the results from worker $2$ and worker $3$ are available.}
    %\Description{Illustration of a $(3,2)$ MDS code. }
    \label{fig:MDS}
\end{figure}
\noindent In general, for an $(N,K)$ MDS code, the data matrix $\mathbf X$ is divided into $K$ equal size sub-matrices $\mathbf X_1,\dots,\mathbf X_K$, for $K\leq N$. Then $N$ encoded matrices $ \tilde{\mathbf X}_1, \dots, \tilde{\mathbf X}_N$ are generated by applying $(N, K)$-MDS code to the sub-matrices. If a systematic MDS code is used for encoding, then $ \tilde{\mathbf X}_i=\mathbf X_i$, for $1\leq i \leq K$. Once any $K$ out of the $N$ results are received from the workers nodes, the master node can decode the final result using these $K$ results. 

\textbf{Lagrange Coded Computing (LCC) \cite{yu2019lagrange}.}  MDS-coded computing can inject redundancy to tolerate stragglers in linear computations.  LCC extends this idea for any polynomial-based computation. 
LCC provides a single framework to tolerate stragglers, Byzantine nodes and to protect privacy against colluding workers.  LCC encodes the dataset into $N$ coded datasets $\tilde{\mathbf X}_1, \tilde{\mathbf X}_2, \cdots, \tilde{\mathbf X}_N \in \mathbb V$, where $N$ is the number of worker nodes, and the $i$-th node computes $f(\tilde{\mathbf X}_i) \in \mathbb U$. Specifically, LCC requires that 
\begin{align}
\label{eqn:N_LCC}
N \geq (K+T-1)\deg f+S+2M+1,
\end{align}
to tolerate $S$ stragglers, $M$ Byzantine workers and to ensure privacy of the dataset against any $T$ colluding workers, where $\deg f$ is the degree of the polynomial $f$. We refer
to this scheme as $(N, K, S, M)$ LCC. The overall computation results in LCC can be then recovered when at least $N-S$ nodes return their computations through a Reed-Solomon decoding approach \cite{berlekamp1966non}. This approach has an encoding complexity of $O(N \log^2(K) \log\log(K) v)$ and results in a decoding complexity of $O((N-S) \log^2 (N-S) \log \log (N-S) u)$, where $v$ is the output size of the encoder and $u$ is the input size of the decoder. That is, the encoding complexity is almost linear in $O(Nv)$ and the decoding complexity is almost linear in $O((N-S)u)$. We observe from  (\ref{eqn:N_LCC}) that handling Byzantine workers are \emph{twice} as costly as stragglers in LCC. Hence, Byzantine node detection is resource-intensive in LCC.  

\textbf{Broader use of coded computing.}  LCC~\cite{yu2019lagrange} provides coded redundancy for any arbitrary multivariate polynomial computations such as general tensor algebraic functions, inner product functions, function computing outer products, and tensor computations. Polynomially coded computing~\cite{yu2017polynomial} can tolerate stragglers in bilinear computations such as Hessian matrix computation. Recent works \cite{so2020scalable, narra2020collage,kosaian2019parity,jahani2020berrut, soleymani2021approxifer} have also demonstrated promising results for extending coded computing to beyond polynomial computations such as deep learning training and inference. 

\subsection{Verifiable Computing}
\label{Subsection:verifiable}
 Verifiable computing is an orthogonal paradigm that has been designed to ensure computational integrity \cite{Freivalds1977ProbabilisticMC}. The basic principle behind verifiable computing is to allow a user to verify whether a compute node has performed the assigned computation correctly. While there are a variety of approaches to verify computations, in this work we adapt the approach proposed in  \cite{Freivalds1977ProbabilisticMC}.  Consider the problem where a user is interested in computing the matrix-vector multiplication $\mathbfsl y=\mathbf X\mathbfsl b$ by offloading it to a worker node, where $\mathbf X \in \mathbb{F}^{m \times d}_q$ and $\mathbfsl b \in \mathbb{F}^{d}_q$. The user chooses a  vector $\mathbfsl r \in \mathbb{F}^{m}_q$ uniformly at random and computes $\mathbfsl s \triangleq \mathbfsl r \mathbf X$ as a private verification key. This verification key generation is done only once. The worker node then returns the outcome $\hat{\mathbfsl y}$, where $\hat{\mathbfsl y} = \mathbf X \mathbfsl b$ if the node performed the computation correctly. The user then performs the following verification check: $\mathbfsl r\cdot \hat{\mathbfsl y} = \mathbfsl s\cdot{\mathbfsl b}$. If the worker node passes this verification check successfully, then the user accepts this computation result. Note that this verification step is done in only $O(m+d)$ arithmetic operations and is much faster compared to computing $\mathbfsl y = \mathbf X {\mathbfsl b}$ on the server, which incurs a complexity of $O(md)$. That is, through this step, the user performs substantially fewer operations compared to the original computation to identify any discrepancy in the computations with high probability.
 
%
%\name uses a combination of dynamically adapted coded computing and verifiable computing to achieve its goals as we describe next. 

\subsection{Related Coding Strategies  for Byzantine Detection}
\label{Subsection: Related Work}

Numerous works  considered the problem of tolerating Byzantine workers in distributed computing and learning settings. For instance, Draco~\cite{chen2018draco} and Detox~\cite{rajput2019detox} introduce algorithmic redundancy that tackles the Byzantine problem in isolation without dealing with stragglers.  Coding-theoretic approaches beyond LCC and Polynomial codes have been proposed recently to tolerate Byzantine nodes in distributed settings \cite{subramaniam2019collaborative,yang2021coded,soleymani2021list}. In \cite{yang2021coded},  two schemes have been proposed to improve the adversarial toleration threshold of LCC based on decomposing the polynomials as a series of monomials.  This decreases the effective degree of the polynomial, but increases the communication cost significantly.   In \cite{subramaniam2019collaborative}, this problem was also considered for the Gaussian and the uniform \emph{random} error models. Our work, however, considers the worst-case error model. 
More recently, a list-decoding approach with side information has been developed in \cite{soleymani2021list} which uses the folding technique in algebraic coding to tolerate more Byzantine nodes. This approach, however, has a significant decoding complexity that is quadratic in the number of workers while it is almost linear in LCC.

Another line of work focused on detecting Byzantine behavior through verifiable computing. An information-theoretic verifiable computing scheme for polynomial computations was proposed in \cite{sahraei2019interpol,sahraei2020infocommit} in the single user-server setting based on Freivalds' algorithm. Leveraging verifiable computing in machine learning has also been considered  in many works such as  \cite{ghodsi2017safetynets,Tramr2019SlalomFV,ali2020polynomial}. In particular, Slalom~\cite{Tramr2019SlalomFV} uses Trusted Execution Environment (TEE)-GPU collaboration for privacy-preserving and verifiable inference. Because of the hardware limitations of TEE and consequently its lower performance, Slalom offloads the linear operations to an untrusted GPU and uses Freivalds'  algorithm for verification~\cite{Freivalds1977ProbabilisticMC} which is less compute-intensive than the original computation. However, this scheme is designed for a single GPU system and cannot scale to the distributed system. Also, it does not support training and it does not mitigate stragglers. 

\subsection{Contributions}
\label{Subsection: Contributions}
The question we pose in this work is whether there is a way to exploit verifiable computing in conjunction with coded computing to get the best of both worlds. That is to use coded computing to tolerate stragglers and ensure data privacy, while using verifiable computing to tolerate Byzantine workers. Such a decoupled approach will lower the cost of tolerating Byzantine workers compared to LCC. We consider a general scenario in which the computation is distributed across $N$ nodes, and propose \emph{Adaptive Verifiable Coded Computing} (AVCC), a new framework to simultaneously mitigate stragglers, provide security against Byzantine workers and provide data privacy. Unlike LCC, AVCC only requires that
\begin{align}
\label{eqn:N_AVCC}
N \geq (K+T-1)\deg f+S+M+1.
\end{align}

\noindent Compared to (\ref{eqn:N_LCC}) in LCC, note that in (\ref{eqn:N_AVCC}) the cost of tolerating a Byzantine node is the same as that of a straggler node. Hence, instead of the $2M$ nodes required in LCC, AVCC only needs $M$ nodes, as we demonstrate in Section \ref{sec:implementation}.   

The key idea of AVCC is to use separate mechanisms for mitigating straggler effects and for tackling Byzantine nodes.  AVCC ensures the computational integrity by verifying the computation of each node independently using node's own compute results until it gets the minimum number of verified results required for decoding. This verification step can start as soon as the first node responds, unlike LCC which requires $N-S$ workers to respond before starting to decode. Hence, AVCC is also highly parallelizable. 

In principle, AVCC can be applied to the distributed computation of any polynomial $f$. However, AVCC is particularly suitable for matrix-vector, matrix-matrix multiplications and in machine learning applications such as linear regression and logistic regression, as verifying such computations is highly efficient \cite{Freivalds1977ProbabilisticMC}.

% \begin{comment}
% \begin{table*}[htb]
% \caption{Comparison of features and applications of various prior techniques}
% %\vskip -0.1mm
% \label{tab:background}
% \resizebox{\textwidth}{!}{%
% \begin{tabular}{lcccccccccccc}
% 	\hline
% 	\hline
% 	\textbf{Method} & \textbf{Straggler Mitigation} & \textbf{Byzantine Detection} & \textbf{Applications} & \textbf{Computation Overhead} & \textbf{Communication Overhead}\\
% 	%%%%%%%%%
%     \midrule
% 	%%%%%%%%from quo
% 	Slalom~\cite{Tramr2019SlalomFV} &$\circ$ & $\bullet & DNN inference	&  $O()$ & $O()$ \\
% 	Draco~\cite{chen2018draco} &$\circ$ & $\bullet & DNN Training	&  $O()$ & $O()$ \\
% 	Detox~\cite{rajput2019detox} &$\circ$ & $\bullet & DNN Training	&  $O()$ & $O()$ \\
% 	%%%%%%%%from quo
% 	%ABY3~\cite{mohassel2018aby3} &$\bullet$ & $\circ$& $\circ$	 &  $\bullet$ & $\circ$ & $\circ$& $\bullet$&&& $\circ$\\
% 	%%%%%%%%
% 	%SecureML~\cite{mohassel2017secureml} &$\bullet$ & $\bullet$ & $\circ$	& $\bullet$ & $\bullet$&$\circ$&$\bullet$ &&& $\circ$\\
% 	%Chiron~\cite{hunt2018chiron} &$\bullet$ & $\bullet$ & $\circ$	& $\circ$ & $\circ$&$\bullet$& $\bullet$&$\bullet$&$\bullet$& $\bullet$& $\circ$&$\circ$\\
% 	%MSP~\cite{hynes2018efficient} &$\bullet$ &$\bullet$  & $\circ$	& $\circ$ &$\circ$ &$\bullet$& $\bullet$&$\bullet$&$\bullet$ &$\bullet$&$\circ$&$\circ$ \\
	
% 	\hline
% \end{tabular}
% }
% \end{table*}
% \end{comment}

%% file: SystemStructure.tex
\section{Problem Setting}
\label{sec:model}
In this section, we describe our setting, the threat model and our guarantees. 
\subsection{System Setting and Threat Model}
We consider a distributed system with $N$ nodes and a main server, denoted as master, that distributes the data among the worker nodes. Given a dataset $\mathbf X=[\mathbf X_1^\top, \mathbf X_2^\top, \cdots, \mathbf X_K^\top]^\top$, our final goal is to compute $\mathbf Y_i=f(\mathbf X_i), \forall i \in [K]$. To this end, the main server first encodes that data into $N$ coded datasets denoted by $\tilde{\mathbf X}_1, \tilde{\mathbf X}_2, \cdots, \tilde{\mathbf X}_N$. The $i$-th  worker then receives $\tilde{\mathbf X}_i$,  computes $\tilde{\mathbf Y}_i=f(\tilde{\mathbf X}_i)$ and sends the result back to the main server for verification and decoding. The main server collects all computations from the non-straggling workers, first verifies them as we explain in the following sections and finally recovers the computation outcome $\mathbf Y_1, \cdots, \mathbf Y_K$ from the fastest workers that passed the verification.

We assume that the system has up to $S$ stragglers that have higher latency compared to the other workers. While the main server is trusted, the worker nodes can be dynamically malicious. Therefore, adversaries on the workers can have full control (root access) and design any attack. As a result, at any given time, some of the worker nodes can send arbitrary results to the main server to sabotage the computation. In addition, some workers may send incorrect computations unintentionally. Specifically, we assume that up to $M$ worker nodes can return erroneous computations with no limitation on their computational power. Finally, we assume that up to $T$ curious workers can collude to learn information about the dataset. 
%Furthermore, we data privacy must be ensured. 
\subsection{Guarantees}
Our goal is to design a scheme that provides the following guarantees.  
\begin{enumerate}
\item \textbf{$S$-Resiliency}. The computation outcome must be recovered even in the presence of $S$ stragglers. 
\item \textbf{$M$-Security}. This means that the system can tolerate up to $M$ workers sending erroneous computations, with no limitations on their computational capability, with an arbitrarily high probability based on verification overhead. 
\item \textbf{$T$-Privacy}. The workers must remain oblivious to the dataset in the information-theoretic sense even if $T$ of them collude. That is, for every set of at most $T$ workers denoted by $\mathcal T \subset [N]$, we must have
\begin{align}
    I(\mathbf X; \tilde{\mathbf X}_{\mathcal T})=0,
\end{align}
where $I(. ; .)$ denotes the mutual information~\cite{jaynes1957information} and $\tilde{\mathbf X}_{\mathcal T}$ denotes the encoded datasets at the workers in $\mathcal T$. 
%Thus the mutual information is zero. 
\end{enumerate}

%% file: implementation.tex
\section{Adaptive Verifiable Coded Computing (AVCC)}
\label{sec:implementation}

In this section, we present our adaptive verifiable coded computing (AVCC) framework, which consists of five key components: $1)$ Data encoding; $2)$ Verification key generation; $3)$ Integrity check; $4)$ Decoding; $5)$ Dynamic coding. We start with an example to illustrate the key components of AVCC.
%In Section \ref{subsec:data encoding}, we explain how the data is encoded in VCC to tolerate stragglers and to protect the privacy of the data. Next, in Section \ref{subsec: integrity check}, we illustrate the integrity check of VCC through which the master verifies the individual computation of each worker to detect the malicious workers. Finally, in Section \ref{subsec:dynamic coding}, we provide our dynamic coding strategy. 

\subsection{Illustrating Example}
\label{example1}
We focus on the logistic regression problem to illustrate how AVCC works. Given a dataset $\mathbf X \in \mathbb R^{m \times d}$ of $m$ data points and $d$ features and a label vector $\mathbfsl y \in \mathbb R^m$, the goal in logistic regression training  is to find the weight vector $\mathbfsl w \in \mathbb R^d$ that minimizes the cross entropy function
\begin{align}
C(\mathbfsl w)=\frac{1}{m}\sum_{i=1}^{m}(-y_{i}\log\hat{y_{i}}-(1-y_{i})\log(1-\hat{y_{i}})),
\end{align}
where $\hat{y_{i}}=h(\mathbfsl x_{i}\cdot\mathbfsl w) \in (0,1)$ is the estimated probability of label $i$ being equal to 1, $\mathbfsl x_{i}$ is the $i$-th row of $\mathbf X$, and $h(\cdot)$ is the sigmoid function $h(\theta)=1/(1+e^{-\theta})$. The gradient descent algorithm solves this problem iteratively by updating the model as 
\begin{align}
\mathbfsl w^{(t+1)}=\mathbfsl w^{(t)}-\frac{\eta}{m} \mathbf X^{\top} (h(\mathbf X\mathbfsl w^{(t)}) -\mathbfsl y),
\end{align}
where $\eta$ is the learning rate and the function $h(\cdot)$ operates element-wise on the vector $\mathbf X\mathbfsl w^{(t)}$. In this example, we provide a two-round protocol as follows.\footnote{The dataset and the weight vector at each round are quantized and represented over the finite field to guarantee information-theoretic privacy \cite{so2021codedprivateml}.} 

In the first round,  an intermediate vector $\mathbfsl z^{(t)} = \mathbf X \mathbfsl w^{(t)}$ is computed, which is then used to compute the predicted probability $h(\mathbfsl z^{(t)})$ and the prediction error vector $\mathbfsl e^{(t)} = h(\mathbfsl z^{(t)}) - \mathbfsl y$. In the second round, the gradient vector $\mathbfsl g^{(t)} = \mathbf X^{\top} \mathbfsl e^{(t)}$ is computed. 

We now illustrate how  to compute this logistic regression task in a distributed setting in the presence of stragglers and compromised compute nodes as depicted in Fig. \ref{fig:Overview}.
\begin{figure*}[t]
    \centering
    \includegraphics[scale=0.42, trim={2cm 0.4cm 4.5cm 0},clip]{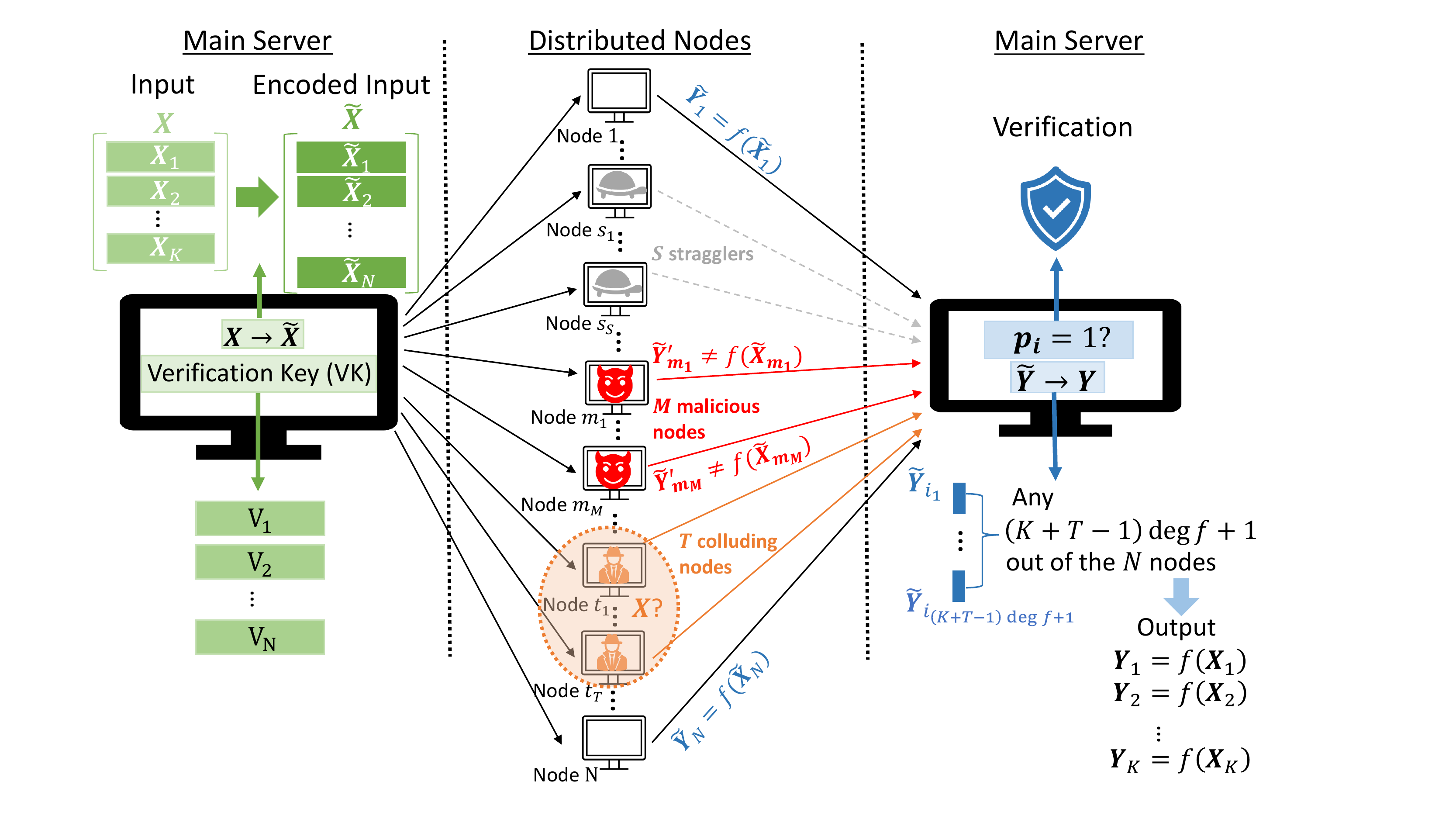}
    \caption{An overview of the Adaptive Verifiable Coded Computing (AVCC) framework is shown. In AVCC, the main server (master)  verifies the computation of each worker individually as soon as this worker sends its computation result using the initially computed verification keys. The main server then reconstructs the final output using the results of the fastest and verified workers.}
    %\Description{Overview of the Adaptive Verifiable Coded Computing (AVCC) framework.}
    \label{fig:Overview}
\end{figure*}

\begin{enumerate}
    \item \textbf{Data Encoding.} Before starting the computation, the main server partitions the dataset $\mathbf X$ into $K$ sub-matrices and encodes them using $(N,K)$-MDS coding. As stated earlier, MDS encoding is a special case of LCC encoding when the computations are only linear.  The main server then sends the coded sub-matrix $\tilde{\mathbf X}_i$ to the $i$-th worker, where $i \in [N]$.
    \item \textbf{Verification Key Generation.} The main server also computes a one-time verification key that helps in verifying the computation results returned by the worker nodes afterwards. Specifically, the main server chooses a  vector $\mathbfsl r_i^{(1)} \in \mathbb{F}^{m/K}_q$ and a vector $\mathbfsl r_i^{(2)} \in \mathbb{F}^{d/K}_q$ uniformly at random for each worker $i \in [N]$.
    
   The main server then computes the following private verification keys as in the Freivalds' algorithm \cite{Freivalds1977ProbabilisticMC} as follows 
    \begin{align}
    &\mathbfsl s_i^{(1)} \triangleq \mathbfsl r_i^{(1)} \tilde{\mathbf X}_i, \\
    &\mathbfsl s_i^{(2)} \triangleq \mathbfsl r_i^{(2)} \tilde{\mathbf X}_i^{\top }.
    \end{align}
    The main server then keeps these private verification keys for the two rounds along with $\mathbfsl r_i^{(1)}$ and $\mathbfsl r_i^{(2)}, \forall i \in [N]$. 
\item \textbf{Integrity Check.} At the first round of iteration $t$, when the $i$-th worker returns its result $\tilde{\mathbfsl z}_i^{(t)}$, which is $\tilde{\mathbf X}_i \mathbfsl w^{(t)}$ if this worker is not Byzantine, the main server checks the following equality 
\begin{align}
\mathbfsl s_i^{(1)} \cdot \mathbfsl w^{(t)} = \mathbfsl r_i^{(1)} \cdot \tilde{\mathbfsl z}^{(t)}_i.
\end{align}
At the second round of iteration $t$, when the $i$-th worker returns $\tilde{\mathbfsl g}^{(t)}_i$ which is $\tilde{\mathbf X}_i^\top \mathbfsl e^{(t)}$ if this worker is not Byzantine, the main server checks the equality 
\begin{align}
\mathbfsl s_i^{(2)} \cdot \mathbfsl e^{(t)} = \mathbfsl r_i^{(2)} \cdot \tilde{\mathbfsl g}^{(t)}_i.
\end{align}
These verification steps are done in only $O(m+d)$ arithmetic operations and are much faster compared to computing $\tilde{\mathbf z}^{(t)}_i$ and $\tilde{\mathbfsl g}^{(t)}_i$ which requires $O(\frac{m}{K}d)$ arithmetic operations. Through these verification steps, the main server can identify the Byzantine workers with high probability that depends on the finite field size $q$ \cite{sahraei2019interpol} as follows
\begin{align}
&\Pr[\mathbfsl r_i^{(1)}\cdot \tilde{\mathbfsl z}_i^{(t)} = \mathbfsl r_i^{(1)}\cdot \tilde{\mathbfsl z}'_i ] \leq \frac{1}{q}, \\
&\Pr[\mathbfsl r_i^{(2)}\cdot \tilde{\mathbfsl g}_i^{(t)} = \mathbfsl r_i^{(2)}\cdot \tilde{\mathbfsl g}'_i] \leq \frac{1}{q},
\end{align}
for any $\tilde{\mathbfsl z}'_i \ne \tilde{\mathbfsl z}_i^{(t)}$ and $\tilde{\mathbfsl g}'_i \ne \tilde{\mathbfsl g}_i^{(t)}$ . 
\item \textbf{Decoding.} The main server decodes the results in each round using the MDS decoding process with the additional constraint that each of the $K$ results it uses has been verified first as explained in step $3$. 

The decoding starts when the main server collects $K$ verified results. Following the property of MDS coding that any $K\times K$ sub-matrix formed by any $K$ columns of the $K\times N$ encoding matrix is invertible, the decoding algorithm is simply to multiply the result matrix formed by concatenating the returned results from $K$ verified workers, with the inverse of the matrix formed by the $K$ columns of the encoding matrix corresponding to the indices of the $K$ verified workers. Thus, the main server can decode and recover the final output using $K$ verified results, instead of just using the first $K$ results. For Byzantine workers that fail the verification, they are effectively treated as stragglers and their results are not used for decoding. %VCC enables fast completion of computation when results from up to $N-K$ nodes are not ready, due to stragglers or malicious nodes. 

\item \textbf{Dynamic Coding.} Ignoring a Byzantine worker's result comes at the cost of reduced straggler tolerance as the main server has to wait for an additional verified result. In the original MDS coding strategy, $N-K$ stragglers can be tolerated.  If the original MDS straggler tolerance is still desired, the main server changes the coding strategy to $(N-1, K-1)$ to account for this Byzantine worker. In this coding strategy, each worker now performs more work but only $K-1$ results are needed to decode. Thus, this dynamic coding approach enables the main server to trade-off redundant work with Byzantine tolerance. 
\end{enumerate}

\begin{remark}(\name Improvements over LCC).
\name has two key improvements over LCC. First, \name relies on LCC for straggler tolerance and ensuring privacy only but it depends on verifiable computing to identify the Byzantine workers. This verification process can be done independently for each worker node without waiting for the results of other workers. LCC, in contrast, has to wait for the results of a sufficient number of workers before identifying the Byzantine workers.  This allows \name to use less number of workers compared to LCC. Second, \name uses a dynamic coding approach to automatically trade-off straggler tolerance for higher Byzantine protection and vice-versa.
\end{remark}
\subsection{Generalized \name}
\label{AVCCdescr}
As explained in Section \ref{example1}, AVCC is particularly suitable for matrix-vector and matrix-matrix
multiplications as the information-theoretic verification schemes of such computations are efficient \cite{sahraei2019interpol}. Such computations are essential in several machine learning applications such as linear regression and logistic regression. However, in principle, AVCC can be applied to any polynomial $f$. We now explain the encoding, the verification and the decoding of AVCC. 
\begin{enumerate}
\item \textbf{Data Encoding.} In AVCC, the dataset $\mathbf X=(\mathbf X_1^\top, \mathbf X_2^\top, \cdots, \mathbf X_K^\top)^\top$ is encoded as in LCC, but AVCC requires less number of workers. Specifically, the encoding is as follows. First, a set of $K+T$ distinct elements denoted by $\mathcal B=\{\beta_1, \cdots, \beta_{K+T}\}$ are chosen from $\mathbb F_q$ and an encoding polynomial of degree at most $K+T-1$ is constructed such that $u(\beta_j)=\mathbf X_j$ for $j \in [K]$ and $u(\beta_j)=\mathbf W_j$ for $j \in \{K+1, \cdots, K+T\}$, where $\mathbf W_j$ is chosen uniformly at random. Such a polynomial can be constructed as follows
\begin{align}
    u(z)=\sum\limits_{j=1}^{K} \mathbf X_j \ell_j(z)+ \sum\limits_{j=K+1}^{K+T} \mathbf W_j \ell_j(z),
\end{align}
where $\ell_j(z)$ is the Lagrange monomial defined as follows  \begin{align}
    \ell_j(z)=\prod_{k \in [K+T] \setminus \{j\} } \frac{z-\beta_k}{\beta_j-\beta_k},
\end{align}
for $j \in [K+T]$. Next, another set of $N$ distinct points denoted by $\mathcal A=\{\alpha_1, \cdots, \alpha_N\}$ is selected from $\mathbb F_q$. If $T>0$, then the sets are selected such that $\mathcal A \cap \mathcal B=\emptyset$. The main server then sends $\tilde{\mathbf X}_i=u(\alpha_i)$ to the $i$-th worker which is required to compute  $f(\tilde{\mathbf X}_i)=f(u(\alpha_i))$, where we note that 
\begin{align}
\deg f(u(z)) \leq (K+T-1) \deg f. 
\end{align}

% In the above encoding process it does not consider Byzantine node tolerance. The encoding process only handles stragglers and collusion to protect privacy. In particular, there is no Byzantine tolerance provided during the data encoding process. 
The main difference between the encoding in LCC and the encoding of AVCC is that LCC requires that $N \geq (K+T-1)\deg f+S+2M+1$, whereas AVCC only requires that $N \geq (K+T-1)\deg f+S+M+1$.

It is also worth noting that when $T=0$ and $\deg f=1$ in AVCC, we can encode the dataset using an $(N, K)$ MDS code as illustrated in Subsection \ref{example1}.

\item \textbf{Verification Key Generation.} The main server also computes a random private verification key $\mathbf V_i$ for each worker  $i$ which depends on $\tilde{\mathbf X}_i$ and the polynomial $f$. 
%This verification key is used later to check the correctness of the returned computations. 
\item \textbf{Integrity Check}. When the main server receives the result of the $i$-th worker, it verifies the correctness of this result using the private verification key $\mathbf V_i$. We denote the binary output of the verification algorithm for the $i$-th worker by $p_i$, where $p_i=1$ if this worker passes the verification test and $p_i=0$ otherwise. 

If the $i$-th worker returns the correct computation result $\tilde{\mathbf Y}_i=f(\tilde{\mathbf X}_i)$, then the main server accepts the result with probability $1$. Otherwise, the main server detects this malicious behavior regardless of the computational power of this worker with high probability \cite{sahraei2019interpol}. Specifically, we have
\begin{align}
\Pr(p_i=0, \tilde{\mathbf Y}_i^{'} \neq f(\tilde{\mathbf X}_i)) \geq 1-o(1),
\end{align}
where the term $o(1)$ is a term that vanishes as the finite field size $q$ grows. 
\item \textbf{Decoding.} Once the main server collects $(K+T-1) \deg f+1$ verified results, it then interpolates the polynomial $f(u(z))$ and reconstructs the computation outcome by evaluating the polynomial $f(u(z))$ at $\beta_i$ $\forall i \in [K]$ as $f(u(\beta_i))=f(\mathbf X_i)$.

\item \textbf{Dynamic Coding.}
The main server may decide to dynamically reconfigure the coded data distribution, based on the observed system behaviour in the previous iteration(s). %This strategy is specifically for the case where $T=0$ and $\deg f=1$.  \\
%After detecting Byzantine workers and stragglers, we need to reconfigure our coding scheme based on the number of available workers. 
Assume our system has $N$ workers and initially uses $(N, K)$ MDS coding, where $N=K+M+S+T$. We claim that the system tolerates up to $S$ stragglers, $M$ Byzantine workers and $T$ colluding workers. Our strategy changes the dimension of the code and the code length dynamically based on the history. We denote the dimension of the code at time $t$ by $K_t$ and the number of workers in the system at time t by $N_t$. That is, we use $(N_t, K_t)$ MDS code at iteration $t$. Suppose that at iteration $t$, we detect $S_t$ stragglers and $M_t$ malicious workers, and there are potentially $T_t$ colluding workers in the system, where $S_t \leq S$, $M_t \leq M$ and $T_t \leq T$. We define a parameter $A_t$ that shows how many additional stragglers we can tolerate in the future iterations before we suffer from the tail latency as follows
\begin{align}
A_t = N_t-M_t-S_t-K_t-T_t.
\end{align}
This parameter determines that the new coding scheme at iteration $t+1$ should be as follows
\begin{align}
(N_{t+1}, K_{t+1}) &=
\begin{cases}
(N_t-M_t, K_t) & \text{if} \  A_t \ \geq 0,\\
(N_t-M_t, K_t+A_t) & \text{if} \ A_t <0.
\end{cases}
\end{align}
That is, our strategy is as follows. If $A_t \geq 0$, we do not have to wait for any stragglers to complete the computation. However, when $A_t <0$ this indicates that we already suffer from stragglers and we need to adapt our coding scheme with the available nodes we have. To do so, we need to reduce the dimension of the code as well as the code length. But re-encoding the data and verification keys based on the new coding scheme can be a performance bottleneck. For this reason, in the preprocessing phase before the application starts, we generate encoded data as well as verification keys of different coding configurations offline. Therefore, we have the flexibility to use them dynamically. 

Similar strategy can be applied with minor modification when the system encodes using Lagrange coding. In the case of using Lagrange coding, we set $A_t$ as follows 
\begin{align}
A_t = N_t-M_t-S_t-(K_t+T_t-1)\deg f,
\end{align}
and the new coding scheme at iteration $t+1$ is as follows
\begin{align}
(N_{t+1}, K_{t+1}) &=
\begin{cases}
(N_t-M_t, K_t) & \text{if} \  A_t \ \geq 0,\\
(N_t-M_t, K_t+\lfloor \frac{A_t}{\deg f}\rfloor) & \text{if} \ A_t <0.
\end{cases}
\end{align}
\end{enumerate}
Theorem \ref{theorem:N_AVCC} characterizes the set of all feasible $S$-resilient, $M$-secure, and $T$-private schemes that AVCC achieves.
\begin{theorem}
\label{theorem:N_AVCC}
Given a number of workers $N$ and a dataset $\mathbf X=(\mathbf X_1^\top, \mathbf X_2^\top, \cdots, \mathbf X_K^\top)^\top$, AVCC provides an $S$-resilient, $M$-secure, and $T$-private scheme for computing ${\{f({\mathbf X}_i)\}}_{i=1}^K$ for any polynomial $f$, as long as
\begin{align}
N \geq (K+T-1)\deg f+S+M+1.
\end{align}
\end{theorem}
\begin{proof}
We start by showing that AVCC is $S$-resilient and $M$-secure and then show it is $T$-private.
\begin{itemize}
\item \textbf{$S$-resiliency and $M$-security}. Since LCC uses Reed-Solomon decoder to identify the Byzantine workers, it requires $2M$ additional workers in order to tolerate $M$ malicious workers. Unlike LCC, AVCC mitigates Byzantine workers by verifying each received worker's result independently. If the worker returns the correct computation result, then the verification algorithm accepts the result with probability $1$. Otherwise, the verification algorithm rejects the erroneous result regardless of the computational power of the worker with a probability at least $1-o(1)$, where the term $o(1)$ is a term that vanishes as the finite field size grows. Hence, AVCC requires only $M$ additional worker results to tolerate $M$ Byzantine workers.\\ Specifically, since AVCC encodes the data the same way as LCC does as described in Section \ref{AVCCdescr}, it follows that AVCC is $S$-resilient and $T$-private as LCC. More specifically, the encoding polynomial $u(z)$ is constructed with $K+T$ distinct points and hence it is of degree at most $K+T-1$. The computation at the worker side is to apply $f$ on the encoded data, that is, to evaluate $f(u(z))$, and the composed polynomial $f(u(z))$ satisfies $\deg f(u(z))\leq (K+T-1)\deg f$. \\
To recover $f(X_i)$, the master first interpolates $f(u(z))$, and then evaluates $f(u(z))$ at $\{\beta_i\}_{i \in [K]}$. To interpolate $f(u(z))$, the master needs a minimum of $\deg f(u(z))+1$ correct evaluations, that is,  $\deg f(u(z))+1$ correct worker results. Since $\deg f(u(z))\leq (K+T-1)\deg f$ and $N \geq (K+T-1)\deg f+S+M+1$, the master is ensured to obtain $\deg f(u(z))+1$ correct worker results and then obtain all coefficients of $f(u(z))$. This is because the master has at least $(K+T-1)\deg f+1$ correct results without the results of the at most $S$ stragglers and after discarding the results of at most $M$ malicious workers.  Hence, we have shown that the AVCC scheme is $S$-resilient and $M$-secure.
\item \textbf{$T$-privacy.} We  recall that AVCC uses the same encoding method as LCC. This encoding can be represented as
\begin{align*}
\tilde{\mathbf X}_i=u(\alpha_i)=(\mathbf X_1, \cdots, \mathbf X_K, \mathbf W_{K+1}, \cdots, \mathbf W_{K+T}). \mathbf U_i,
\end{align*}
where $\mathbf U \in \mathbb F_q^{(K+T) \times N}$ is the encoding matrix, $${\mathbf U}_{i,j} \triangleq \prod_{\mathit k \in [K+T] \setminus \{i\} } \frac{\alpha_j-\beta_k}{\beta_i-\beta_k}$$ and $\mathbf U_i$ is the $i$-th column of $\mathbf U$. Then, it follows from Lemma 2 in \cite{yu2019lagrange} that every $T \times T$ submatrix of the bottom $T \times N$ submatrix $\mathbf U^{\textrm{bottom}}$ of the encoding matrix $\mathbf U$ is invertible. Thus, for every set of $T$ workers denoted by $\mathcal T \subset [N]$, their encoded data $\tilde{\mathbf X}_{\mathcal T}={\mathbf X}\mathbf U_{\mathcal T}^{\textrm{top}}+{\mathbf W}\mathbf U_{\mathcal T}^{\textrm{bottom}}$, where $\mathbf W=(\mathbf W_{K+1},\cdots,\mathbf W_{K+T})$, and $\mathbf U_{\mathcal T}^{\textrm{top}} \in {\mathbb F}_q^{K\times T}, \mathbf U_{\mathcal T}^{\textrm{bottom}}\in {\mathbb F}_q^{T\times T}$ are the top and bottom submatrices which is formed by columns in $\mathbf U$ that are indexed by ${\mathcal T}$. Since $\mathbf U_{\mathcal T}^{\textrm{bottom}}$ is invertible, the added random padding ${\mathbf W}\mathbf U_{\mathcal T}^{\textrm{bottom}}$ is uniformly random. Hence, the coded data ${\mathbf X}\mathbf U_{\mathcal T}^{\textrm{top}}$ is completely masked by the uniformly random mask. This guarantees that the AVCC scheme is $T$-private.
\end{itemize}
\end{proof} 

%% file: Experiments.tex
\section{Experimental setup}
\label{sec:experimentalsetup}
In this section, we describe our experimental setup. We present an empirical study of the performance of AVCC compared to LCC as well as the uncoded baseline. Our focus is on training a logistic regression model for image classification, while the computation load is distributed on multiple nodes on the DCOMP testbed platform \cite{238258}. In our experimental setup, we focus on the case where $T=0$.

We train the logistic regression model described in Section \ref{example1} for binary image classification on the GISETTE dataset \cite{10.5555/2976040.2976109}  to experimentally examine two aspects: the performance gain of AVCC in terms of the training time and accuracy, and the trade-off between various dynamic coding strategies. The size of the GISETTE dataset is $(m,d) = (6000,5000)$. Our experiments are deployed on a cluster of $13$ Minnow instances on a DCOMP testbed, where $1$ node serves as the main server and the remaining $N = 12$ nodes are worker nodes. Each Minnow node is equipped with a quad-core Intel Atom-E processor, $2$GB of RAM and two $1$ GbE network interfaces.

We use $(N,K,S,M) = (12,9,1,1)$ configuration for the LCC baseline in the experiments. For \name we use $(N,K,S+M)=(12,9,3)$. Recall that the encoding process of \name only relies on $N,K$ and hence $S, M$ play no role in the encoding.  %The input matrix is divided into $9$ sub-matrices, and then encoded into $12$ partitions and assigned to the workers in the preprocessing stage. Each worker stores $1$ of the $12$ partitions. In the computation stage, each worker computes the product of its assigned matrix with the vector received from master, and then returns the result. The master then starts to verify the results upon receiving. Once the master collects $9$ verified results, the master starts decoding the result.

We also implement an uncoded baseline which has no redundancy and only $9$ out of the $12$ workers participate in the computation, each of them storing and processing $\frac{1}{9}$ fraction of uncoded rows from the input matrix. The main server waits for all $9$ workers to return, and does not need to perform decoding to recover the result.
\begin{comment}
\begin{algorithm}[htbp]
    \caption{Pseudo-code for quantization}
    \label{alg2}
    \begin{algorithmic}[1]
        \Procedure{Quantization}{\textbf X}
            \State $\mathbf {X}_q$=Field(Round$(\textbf X \cdot 2^l))$ \Comment{$l$ is the quantization parameter}
        \EndProcedure
        \Procedure{ROUND}{\textbf X} 
            \For{$\forall X_i \in \mathbf X$}
                \If{$X_i-\lfloor X_i \rfloor < 0.5$}
                    \State $X_i^r \gets \lfloor X_i \rfloor$
                \Else
                    \State $X_i^r \gets \lfloor X_i \rfloor +1$
                \EndIf
            \EndFor
            \State \textbf{return} ${\mathbf X}^r$
        \EndProcedure
        \Procedure{FIELD}{\textbf X} 
            \For{$\forall X_i \in \mathbf X$}
                \If{$X_i < 0$}
                    \State $X_i \gets X_i + q$ \Comment{two's complement to represent negative numbers in the finite field}
                \EndIf
                \State $X_i^f \gets X_i \bmod q$	
            \EndFor
            \State \textbf{return} ${\mathbf X}^f$	
        \EndProcedure
    \end{algorithmic}
\end{algorithm}\\
\end{comment}

\textbf{Quantization and Parameter Selection.} 
Since the Lagrange coding and the integrity check technique work over a finite field $\mathbb F_q$, but the inputs and the weights for the model training are often defined in real domain, AVCC needs to quantize the inputs and model weights to integers as 
\begin{align}
x_r = \mathrm{round}(2^l \cdot x),
\end{align}
where $x$ represents a floating point number and $l$ is the quantization parameter (number of precision bits). We then embed these integers in the finite field $\mathbb F_q$ of integers modulo a prime $q$. If the integer is negative, we represent it in the finite field using the two's complement representation. When the computation results of the workers are received by the main server, $q$ is subtracted from all the elements larger than $\frac{q-1}{2}$ to restore the negative numbers. The results are then scaled by $2^{-l}$ to convert them back to real numbers. There are many prior works that use quantization schemes in training machine learning models~\cite{gupta2015deep,so2020byzantine, hashemi2021darknight} without noticeable loss in accuracy. 

Matrix multiplication and vector inner product operations are performed in the logistic regression application. Hence, it is necessary to select the field size of each operand to be such that the worst-case computation output still fits within the finite field to avoid a wrap-around which may lead to an overflow error. The Minnow nodes use a $64$-bit implementation, and the number of features in the GISETTE dataset is $d=5000$. Hence, the worst-case operation must satisfy $d(q-1)^2 \leq 2^{63}-1$. As such, we select the finite field size to be $q = 2^{25}-39$ (the largest prime number with 25 bits) to satisfy this limitation. 
%Any model parameter that is larger than $q$ will be quantized. 
%because we perform modular operation on the inner product of vectors, instead of doing it on each element of the product so as to speed up the multiplication. In order to avoid an overflow on intermediate multiplications, $d(q-1)^2 \leq 2^{63}-1$ should hold, and $q = 2^{25}-39$ is the largest prime with $25$ bits that satisfy this requirement. 
 In our experimental setup, the GISETTE dataset values are all non-negative integers and fit within the selected finite field. Hence,  no quantization is necessary. For the model weights, we optimize the quantization parameter to $l=5$ by taking into account the trade-off between the rounding and the overflow error. Note that the bias is implemented as part of the weights, so it shares the same quantization parameter as the weights.

\textbf{Byzantine Attack Models.} We consider the following Byzantine attack models that are widely considered in previous works~\cite{prakash2020mitigating, subramaniam2019collaborative, hashemi2021byzantine}.
\begin{itemize}
    \item \textbf{Reversed Value Attack.} In this attack, the Byzantine workers that were supposed to send $\mathbfsl{z} \in \mathbb F_{q}^{m/K}$ to the main server instead send $-c\mathbfsl{z}$, for some $c>0$. We set $c=1$ in our experiment.
    \item \textbf{Constant Byzantine Attack.} In this model, the Byzantine workers always send a constant vector to the main server with dimension equal to that of the true result.  %and all elements being $-100$ for our experiment.
%    \item \textbf{Uniform Noise Attack.} Byzantine clients draw a random uniform noise from $F_q$ and send it to the master server.
\end{itemize}

\textbf{End-to-end Convergence Performance}. We evaluate the end-to-end convergence performance of AVCC, under two setups and the two attack models, and compare it to LCC and the uncoded approaches. The LCC baseline is designed for ($S=1, M=1$), so it requires $K + S + 2M=12$ workers. AVCC, LCC and the uncoded baseline are all trained for $50$ iterations. For \name we use two different  setups while staying within the constraints that $12 \geq 9 + S + M $ as follows. 
%In order to fulfill the AVCC requirement for $(12,9)$-MDS coding that $12 \geq 9 + S + M $, we select two setups to meet this restriction as follows. 

\begin{enumerate}
    \item $S = 1, M =2$. In this setup, up to $2$ Byzantine nodes may be tolerated but at the expense of reducing the straggler tolerance from a maximum of $3$ nodes to only $1$ node. 
    \item $S=2, M=1$. In the second setup, we reduce the Byzantine tolerance to $1$ node, while reducing straggler tolerance to $2$ nodes. 
\end{enumerate}

\section{Experimental Results}
We now present extensive experimental results showing the performance gain of AVCC over LCC and the uncoded baselines.
\label{sec:experiments}
\begin{figure*}[t]
    \centering
    \begin{minipage}{.25\textwidth}
    \centering
    \includegraphics[scale=0.38]{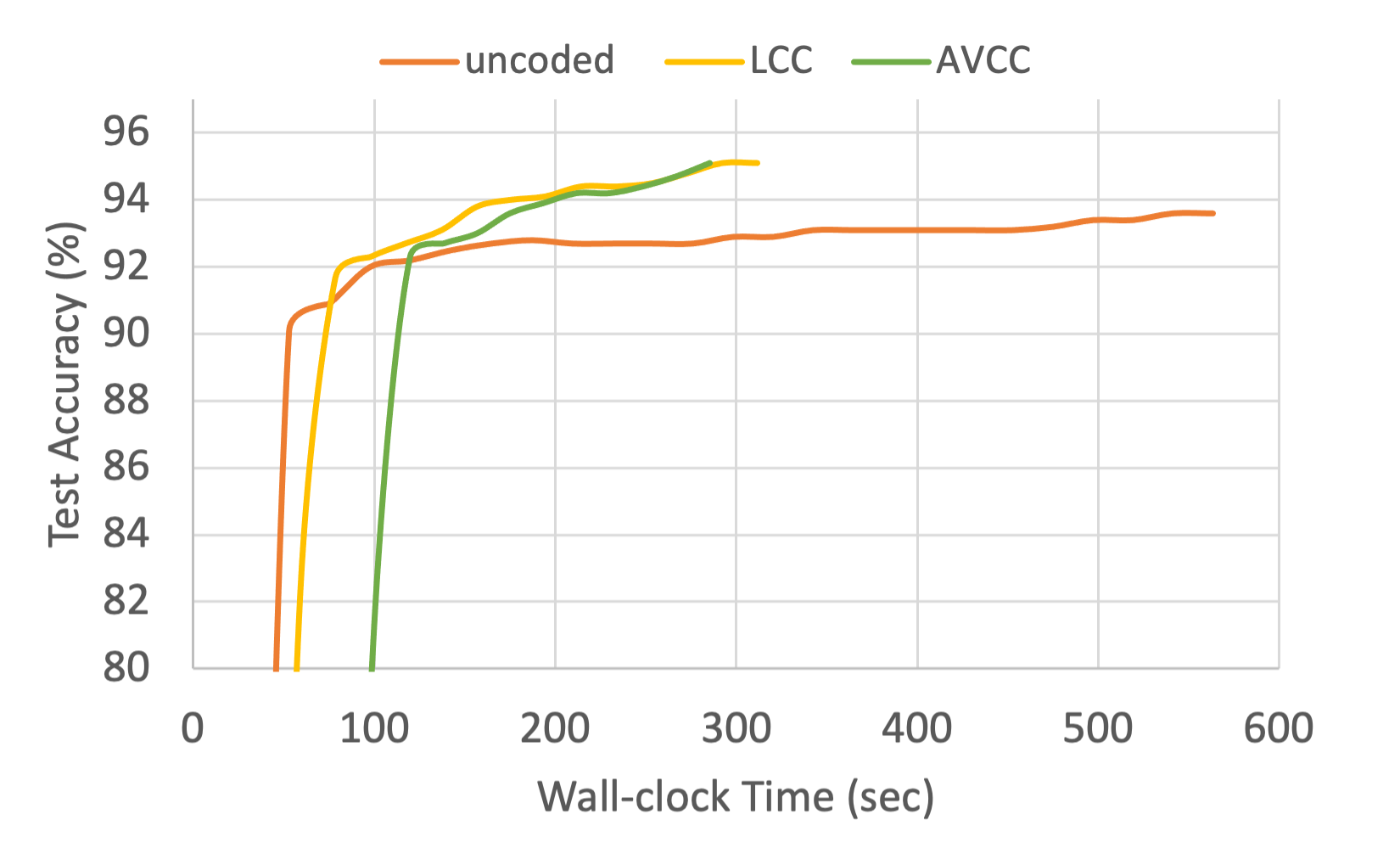}
    \subcaption{Reverse $S=2, M=1$}
    \end{minipage}%
    \begin{minipage}{.25\textwidth}
    \centering
    \includegraphics[scale=0.38]{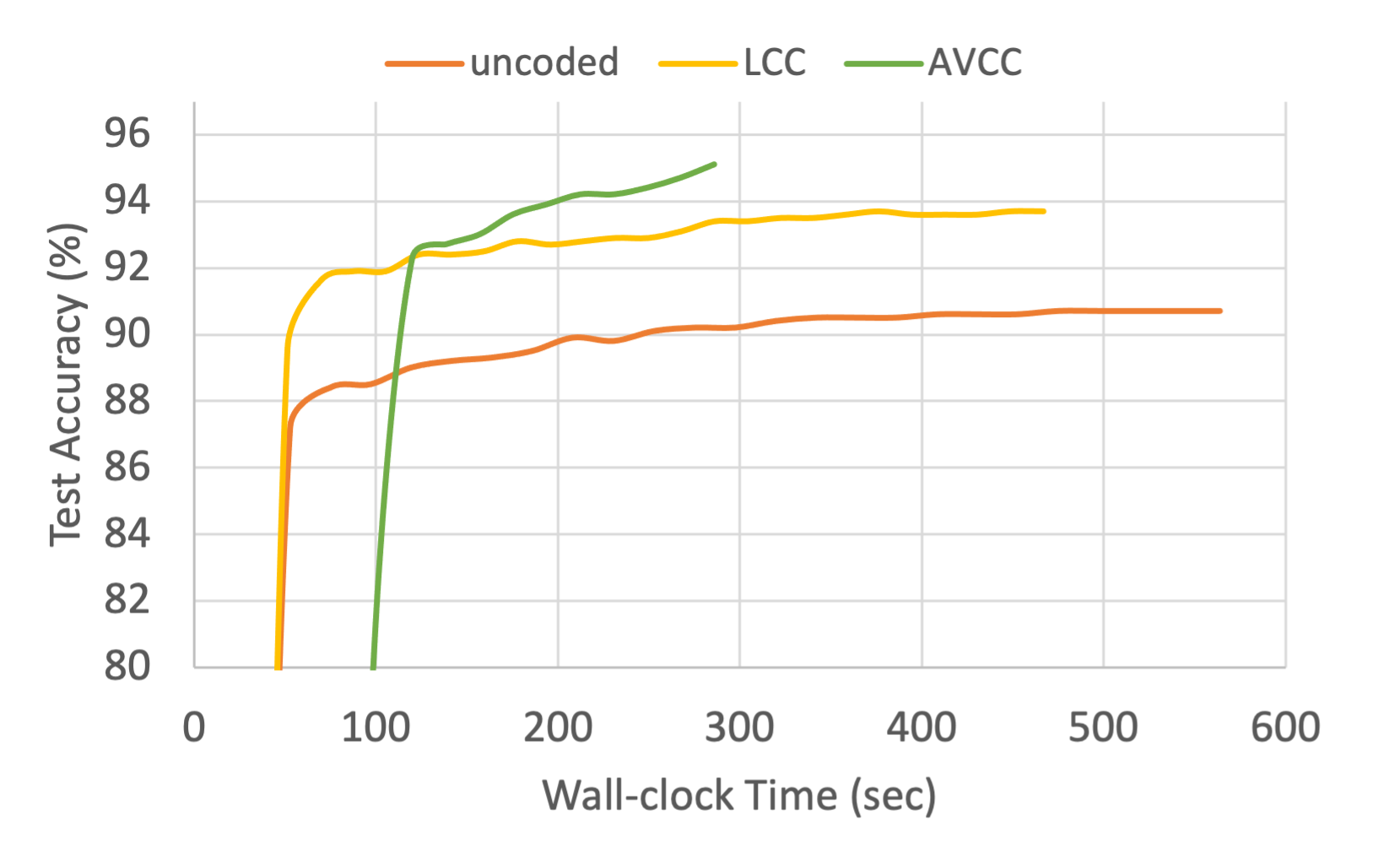}
    \subcaption{Reverse $S=1, M=2$}
    \end{minipage}%
    \begin{minipage}{.25\textwidth}
    \centering
    \includegraphics[scale=0.38]{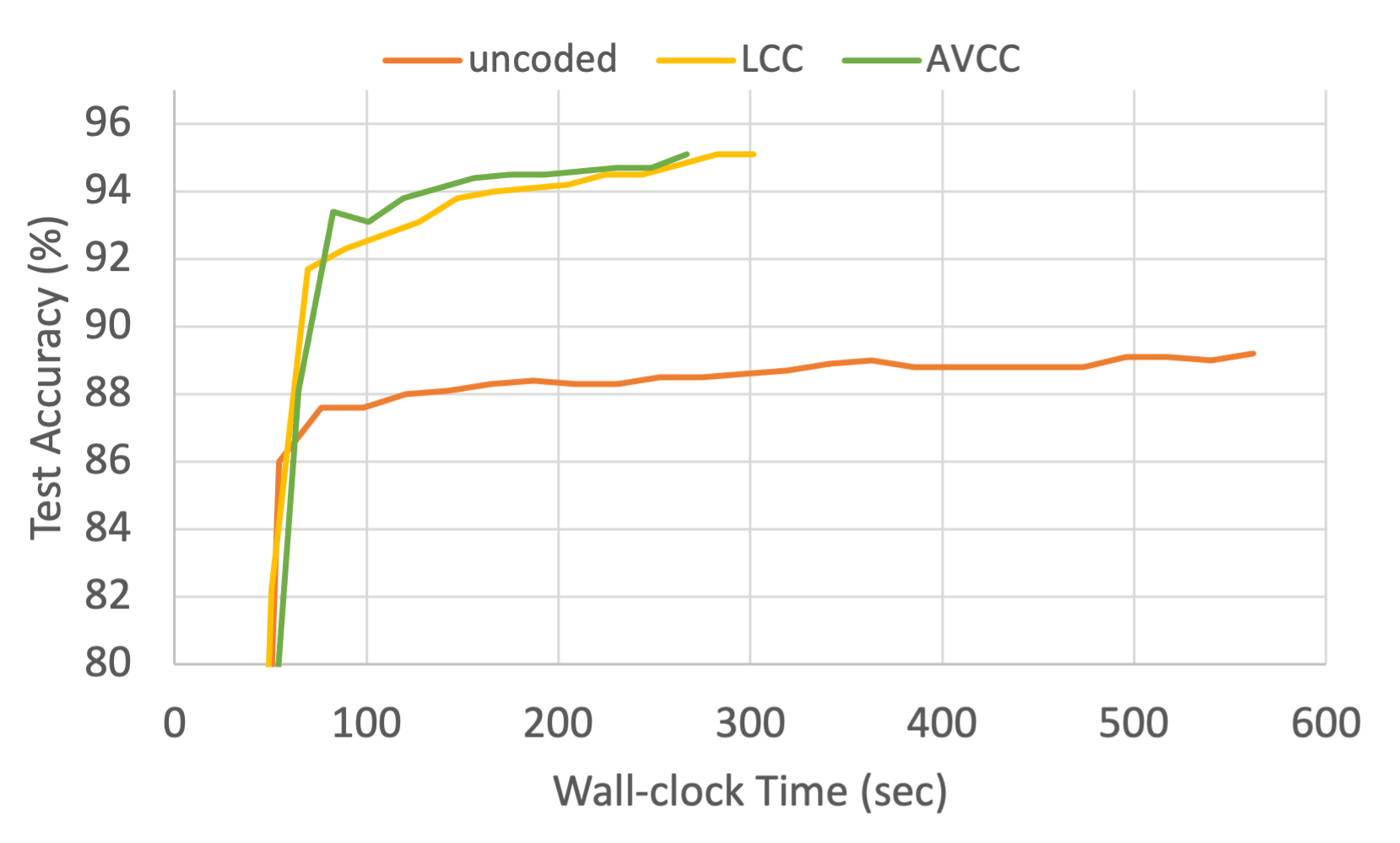}
    \subcaption{Constant $S=2, M=1$}
    \end{minipage}%
    \begin{minipage}{.25\textwidth}
    \centering
    \includegraphics[scale=0.38]{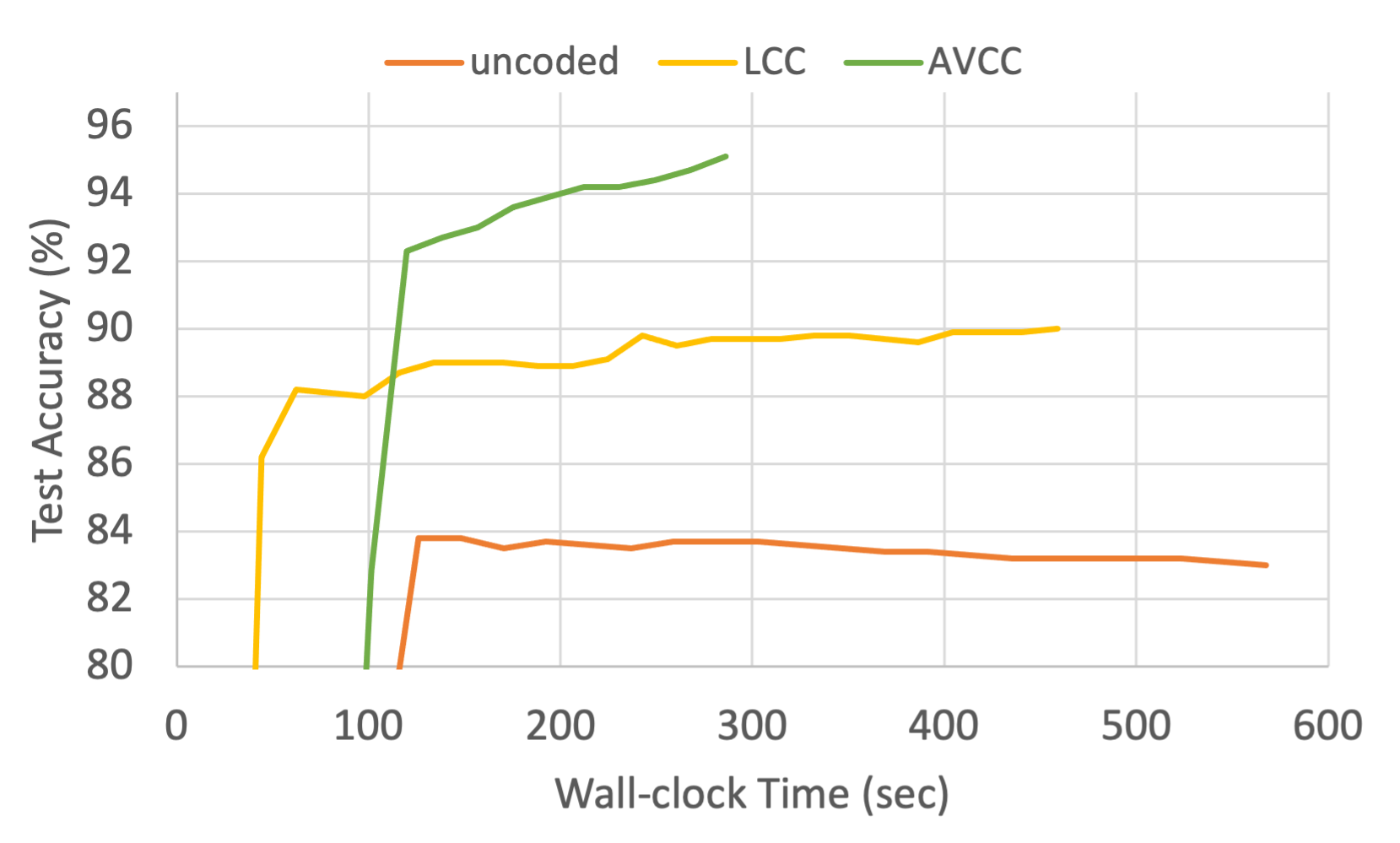}
    \subcaption{Constant $S=1, M=2$}
    \end{minipage}%
    \caption{Convergence performance of AVCC, LCC and uncoded methods with (a) $S=2, M=1$ under reverse value attack, (b) $S=1, M=2$ under reverse value attack, (c) $S=2, M=1$ under constant attack and (d) $S=1, M=2$ under constant attack.}
%\Description{End-to-end comparisons between AVCC and uncoded baseline under constant attack.}
    \label{fig:Attack}
\end{figure*}

\textbf{Accuracy.} We first consider the reverse value attack. This is a weak attack, as the small values produced during the matrix-vector operations when added or subtracted do not derail the overall training convergence  dramatically.

Fig. \ref{fig:Attack}(a) shows the test accuracy under the reverse value attack when there is only one Byzantine node in the system. In this case, since LCC is designed for ($S=1,M=1$), it converges to the same accuracy as AVCC under ($S=2, M=1$) setup, but \name reaches this accuracy level faster than LCC, because LCC by design needs to wait for $11$ worker results before it can start decoding, whereas \name requires only $9$ verified results. Under ($S=2, M=1$) setup, LCC is bound to suffer tail latency from the faster of the two stragglers, while \name in the worst-case can start decoding after verifying results from the $9$ honest workers and rejecting the Byzantine worker's result, without the need to wait for any of the two stragglers.

Fig. \ref{fig:Attack}(b) shows how the test accuracy varies with training time under the reverse value attack when there are two Byzantine nodes in the system.  Recall that LCC is able to handle only one Byzantine node with $N=12$, $K=9$ and $S=1$ by design. In order to handle two Byzantine nodes (while keeping $S=1$), LCC should either increase $N$ to be $14$ or decrease $K$ to be $7$. In the first case, LCC would use more worker nodes (two worker nodes) and correspondingly is more expensive to implement. In the second case, LCC has to assign a lot more work to each worker node, thereby increasing the overall execution latency of each worker. However, \name can dynamically adapt to the changing straggler and Byzantine conditions and automatically adapt to the scenario where there are more Byzantines by trading-off straggler tolerance. Hence, with the ($S=1, M=2$) setup, \name converges to higher accuracy than LCC. 
Note that in both scenarios in Fig. \ref{fig:Attack}(a) and Fig. \ref{fig:Attack} (b), the uncoded scheme does not converge to the same accuracy as \name since it is unable to detect and isolate the Byzantine workers. Hence, the incorrect computations of the Byzantine workers degrade the overall accuracy.  

Fig. \ref{fig:Attack}(c) shows how the test accuracy varies with training time under the constant attack with one Byzantine node. Compared to the reverse value attack, the constant attack is a stronger attack as the malicious behaviour forces all values to a constant value and often causes a considerable accuracy degradation.  LCC is able to detect and isolate the one Byzantine node and hence it reaches an accuracy that is on par with \name.  \name is able to achieve the accuracy at least $10\%$ faster. However, \name shines when the number of Byzantine nodes increases past one. 
\begin{figure*}[h]
    \begin{minipage}[t]{.32\textwidth}
    \includegraphics[width=1\linewidth]{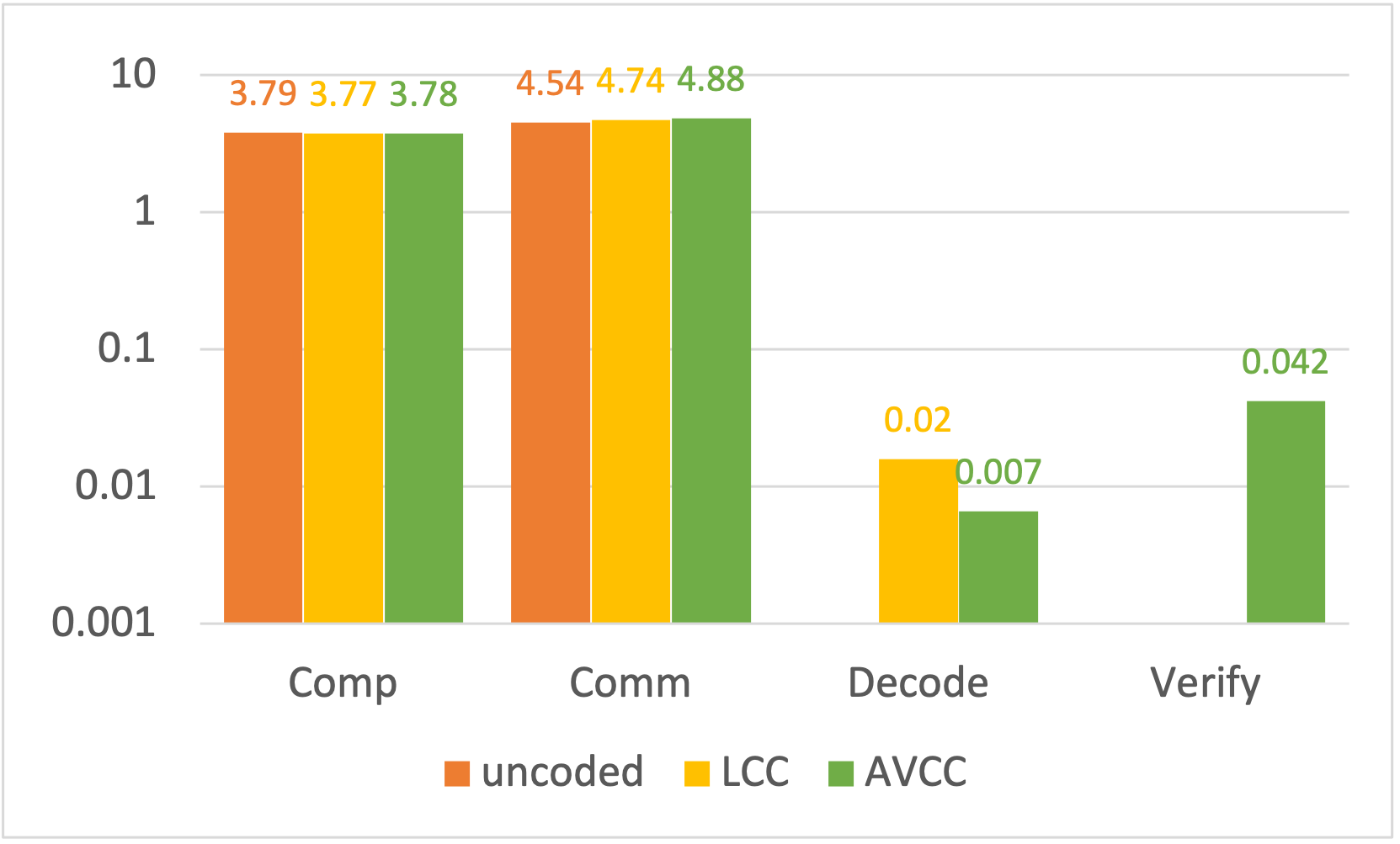}
    \subcaption{$S=0, M=0$}
    \end{minipage}\hfill
    \begin{minipage}[t]{.32\textwidth}
    \centering
    \includegraphics[width=1\linewidth]{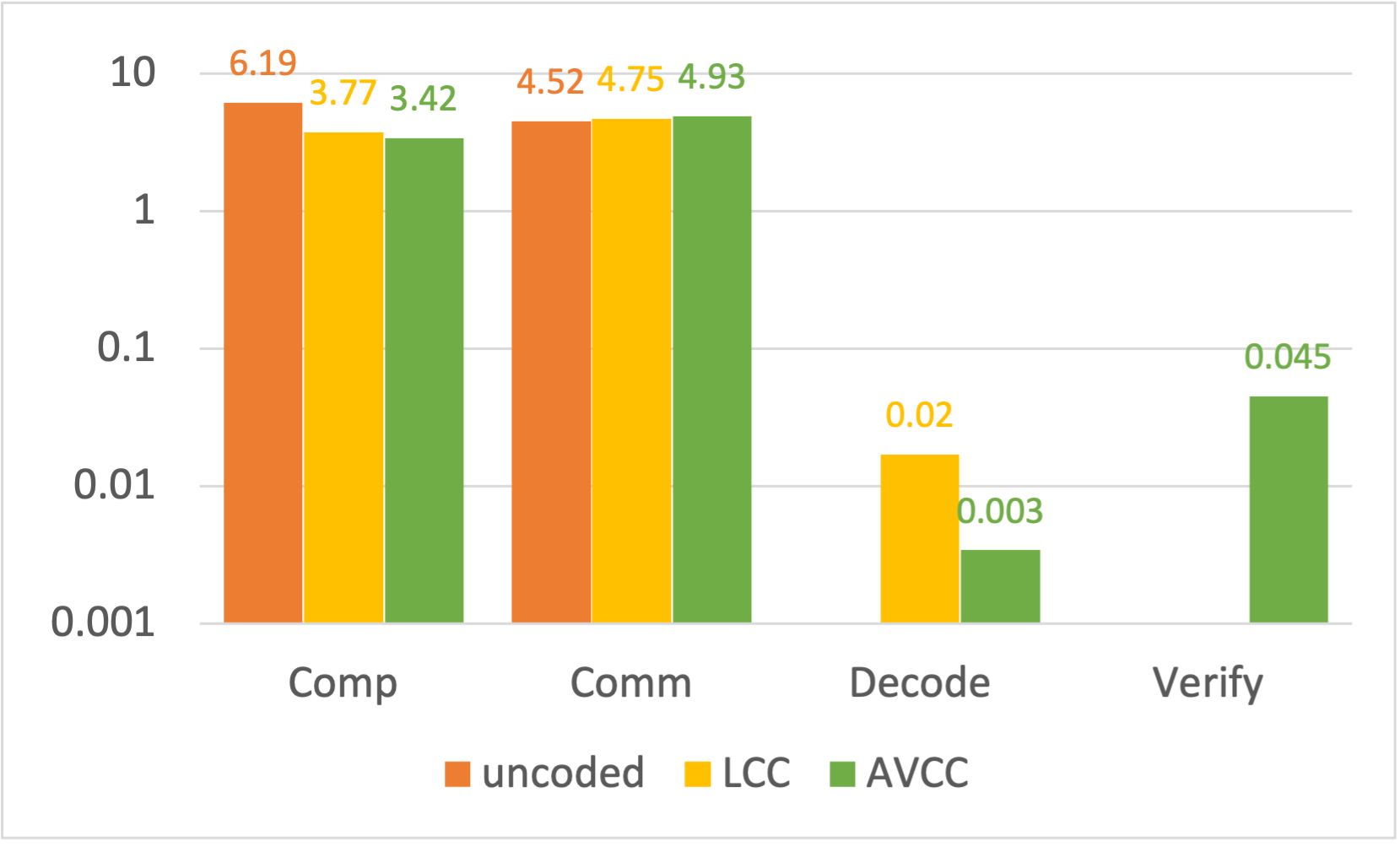}
    \subcaption{$S=1, M=2$ (LCC test accuracy $93.7\%$, AVCC test accuracy $95.1\%$)}
    \end{minipage}\hfill
    \begin{minipage}[t]{.32\textwidth}
    \centering
    \includegraphics[width=1\linewidth]{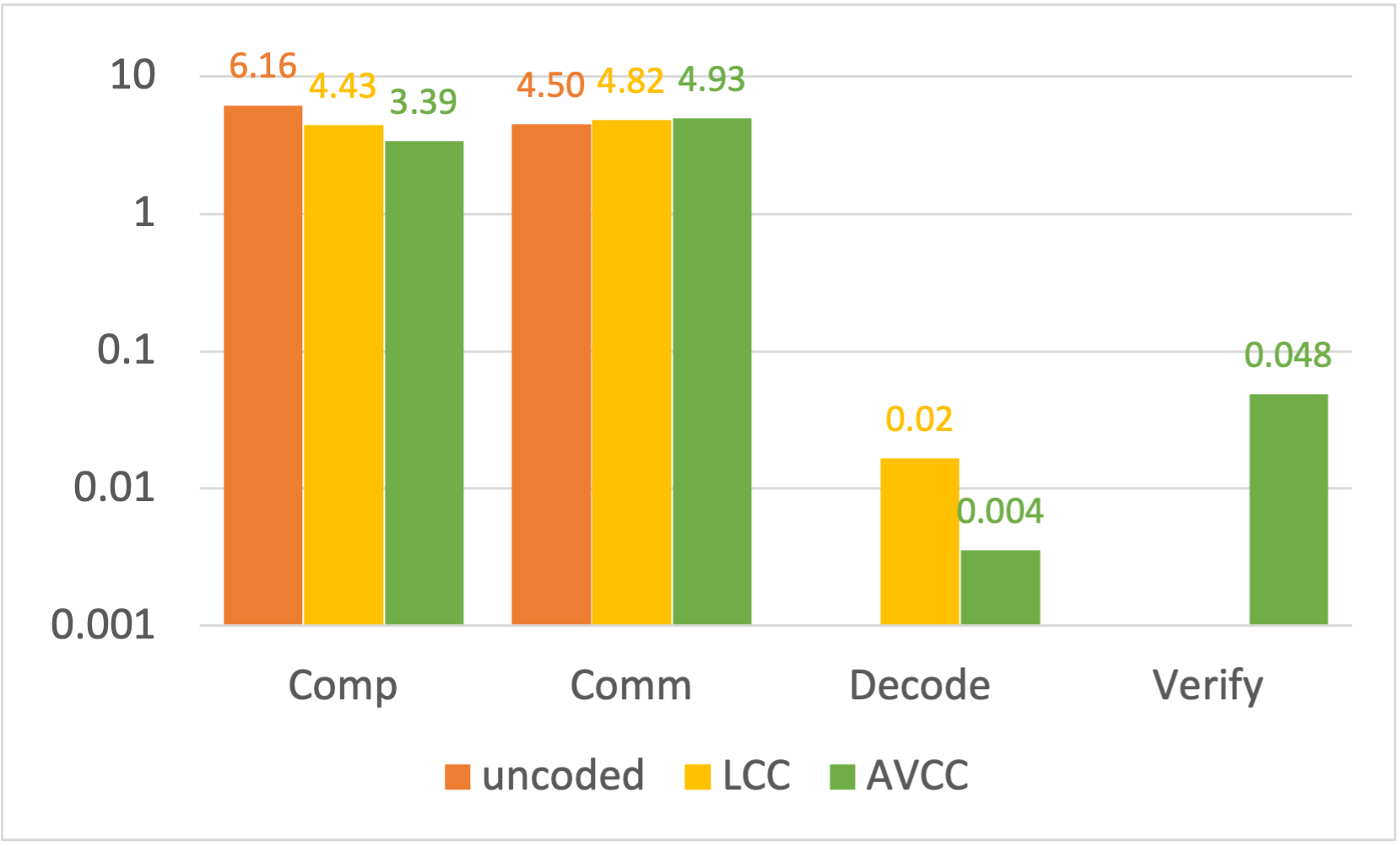}
    \subcaption{$S=2, M=1$ (Test accuracy of both LCC and AVCC $95.1\%$)}
    \end{minipage}\hfill
    \caption{Per iteration runtime analysis of AVCC, LCC and uncoded baseline under different numbers of stragglers and Byzantine nodes. Results under constant attack are similar to that under reverse value attack and thus only results under reverse value attack are shown in (b) and (c).}
    %\Description{Per iteration runtime analysis of AVCC, LCC and uncoded baseline under different number of stragglers and malicious nodes.}
    \label{fig:Profiling}
\end{figure*}

With the $(S=1, M=2)$ setup shown in Fig. \ref{fig:Attack}(d) with the constant attack, LCC converges to $90\%$ accuracy whereas the uncoded scheme converges to a lower accuracy of $83\%$. This result is expected for the uncoded scheme, because with $M=2$ there are two Byzantine nodes in the system that drag the accuracy down when compared to a single Byzantine node setup. As for LCC, it is designed to tolerate ($S=1,M=1$), which protects LCC from one malicious node only. Hence, the accuracy of LCC is lowered by the additional malicious node existing in the system, and converges only to an accuracy of $90\%$.
%a level similar to the uncoded scheme with $M=1$ setup. 
As explained earlier, if the number of Byzantine nodes increases LCC would need either more workers or assign more work per each worker node, both of which are undesirable. 

\begin{comment}
Applying AVCC leads to significant improvement in Byzantine-robustness compared to the other two baselines. The test accuracy increases from $90\%$ to $95.1\%$ as compared to LCC, and from $83\%$ to $95.1\%$ as compared to the uncoded scheme. Under both setups and both attacks, AVCC eventually converges to high accuracy. The reason why LCC with $M=2$ setup as well as the uncoded approach under all setups failed to converge to the same test accuracy as AVCC is that the master effectively lost information about a subset of the input assigned to the Byzantine nodes, and may not recover the desired optimal model. 
\begin{table}[htbp]
\caption{Test accuracy summary.}
\label{table:1}
\begin{center}
\begin{tabular*}{\columnwidth}{c|c|c|c}
\hline
\multicolumn{1}{c|}{} & \multicolumn{1}{c|}{AVCC}   & \multicolumn{1}{c|}{LCC}    & \multicolumn{1}{c}{uncoded} \\ \hline
Reverse value attack\\ S=1, M=2 & 95.1\% & 93.7\% & 90.7\%  \\ \hline
Reverse value attack\\ S=2, M=1 & 95.1\% & 95.1\% & 93.6\%  \\ \hline
Constant attack\\ S=1, M=2      & 95.1\% & 90\%   & 83\%    \\ \hline
Constant attack\\ S=2, M=1      & 95.1\% & 95.1\% & 89.2\%  \\ \hline
\end{tabular*}
\end{center}
\end{table}
\end{comment}

\begin{table}[h]
\centering
\caption{Speedups of AVCC over LCC and the uncoded scheme under various settings.}
\label{table:2}
\begin{tabular*}{.6\columnwidth}{c|c|c}
\hline
\multicolumn{1}{c|}{Setting}           &             \multicolumn{1}{c|}{LCC}    & \multicolumn{1}{c}{Uncoded} \\ 
\hline
\begin{tabular}[c]{@{}c@{}}Reverse value attack\\ $S=1, M=2$\end{tabular} & $2.66\times$ & $5.13\times$  \\ \hline
\begin{tabular}[c]{@{}c@{}}Reverse value attack\\ $S=2, M=1$\end{tabular} & $1.09\times$ & $3.22\times$  \\ \hline
\begin{tabular}[c]{@{}c@{}}Constant attack\\ $S=1, M=2$\end{tabular} & $4.17\times$   & $5.41\times$    \\ \hline
\begin{tabular}[c]{@{}c@{}}Constant attack\\ $S=2, M=1$\end{tabular} & $1.13\times$   & $7.64\times$    \\ \hline
\end{tabular*}
\end{table}

\textbf{Training Time.}  Applying AVCC leads to significant end-to-end speedups as shown in Table \ref{table:2}. In particular, AVCC leads to $4.2\times$ speedup gain over LCC and more than $5\times$ speedup over the uncoded scheme under the constant attack. Under the reverse value attack, AVCC achieves up to $2.7\times$ speedup over LCC and more than $5.1\times$ speedup over the uncoded approach. 

\textbf{Per Iteration Cost of AVCC}. The per iteration cost of applying AVCC to logistic regression using GISETTE dataset is shown in Fig. \ref{fig:Profiling}; note the logarithmic scale on the Y-axis. We breakdown the iteration cost into four categories as follows.
\begin{enumerate}
    \item \textbf{Compute Time}. This is the worst-case latency for performing the matrix operations at any  worker node. 
    \item \textbf{Communication Time}. This accounts for the time to send and receive data between the workers and the main server.
    \item \textbf{Verification Time}. This is the time to verify the results. Note that the cost of encoding and key generation are one-time costs, which get amortized over multiple iterations of the model training or inference.
    \item \textbf{Decoding Time}. This is the decoding time at the main server after the verification.  
    \end{enumerate}
     %For instance, with the Gisette dataset we observed that the encoding cost is 30 seconds and the key verification cost is 
As shown in Fig.~\ref{fig:Profiling}(a), in a straggler-free and Byzantine-free environment, the decoding and verification time of AVCC incurs extra latency. But when there are stragglers in the cluster, the decoding and verification overhead in AVCC is dwarfed by the straggler latency as shown in Fig.~\ref{fig:Profiling}(b) and Fig.~\ref{fig:Profiling}(c). Note that LCC has no separate verification cost since that process is coupled with the decoding process. The presence of stragglers  causes the uncoded execution time to increase substantially. However, both LCC and AVCC are able to tolerate stragglers and Byzantine nodes. and AVCC achieves superior test accuracy even with higher Byzantine node counts.

\textbf{Dynamic Coding}. Recall that AVCC has the ability to dynamically change the coding strategy if straggler or Byzantine nodes persist in the system by re-encoding the data.  We evaluate the benefits provided by re-encoding the data over an approach that performs AVCC functions but without re-encoding the data. We call that approach Static VCC. Static VCC is a constrained version of AVCC, where the verification mechanism is still available to mitigate Byzantine nodes, but the dynamic coding is removed so that the coding scheme will not change throughout the execution.  The primary disadvantage of dynamic re-coding is that newly encoded data must be re-sent to the workers once such a decision is made by AVCC. 

\begin{figure}[t]
    \centering
    \includegraphics[width=.8\linewidth]{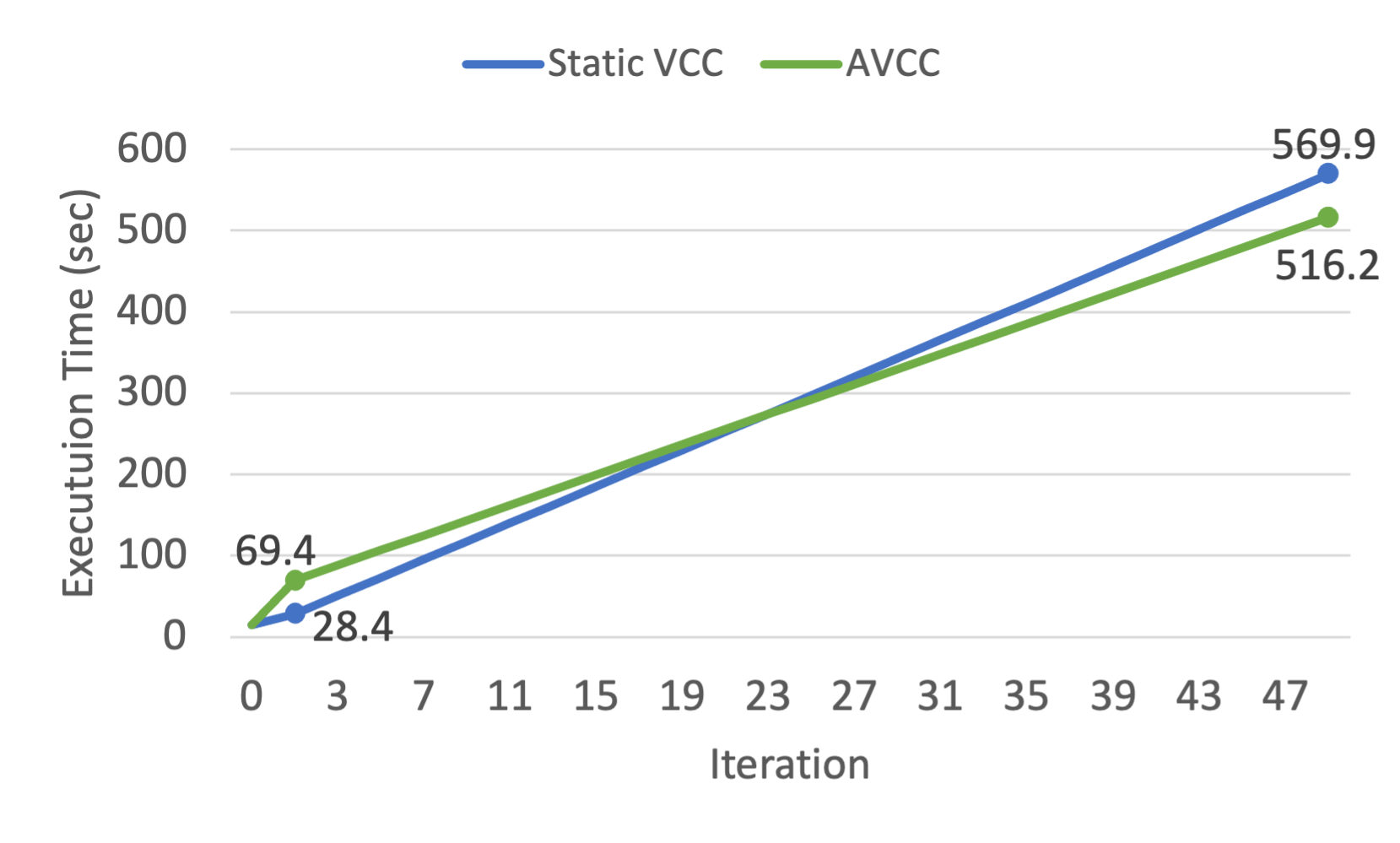}
    \captionof{figure}{AVCC and Static VCC execution time comparison}
    \label{fig:Dynamic2}
    %\Description{Dynamic coding experiment result}
\end{figure}
Fig.~\ref{fig:Dynamic2} quantifies the benefit of AVCC over Static VCC in an exemplary scenario. In this scenario we start with the initial coding scheme  of $(N=12, K=9, S=2, M=1)$. That means we can tolerate one Byzantine node and two stragglers. At the start of iteration $1$, the system encounters three stragglers and one Byzantine node. At this juncture, AVCC can eliminate the one Byzantine node from the group of workers, but it also recognizes that the coding scheme is no longer able to handle 3 stragglers. AVCC then re-encodes the data using $(N=11, K=8, S=3, M=0)$ setting. Static VCC, on the other hand, does not re-encode the data. AVCC incurs a one-time cost of about $41$ seconds at the end of iteration $1$, because it sends the encoded matrices with the updated coding scheme $(11, 8)$ to the workers. In spite of this re-encoding cost, at the end of the $50$ iterations, AVCC saves about $54$ seconds in the overall execution time, compared to Static VCC.  This scenario exemplifies the benefits of the adaptive nature of AVCC. It is also worth noting that the one-time cost of re-encoding and data transfer overhead can be mitigated in different ways. For instance, it is possible for the main server to generate a priori multiple versions of the encoded data and send those versions to the workers, where each version is encoded with a different coding scheme. An alternative scheme would allow the main server to reactively encode the data off the critical path. 

%% file: Conclusion.tex
\section{Conclusion}
\label{sec:conclusion}
In this work, we presented AVCC, a framework for resilient, robust, and private distributed machine learning via coded computing with dynamic coding and verifiable computing. AVCC is robust to up to $M$ malicious nodes, tolerates up to $S$ stragglers, and provides privacy against up to $T$ colluding nodes, while being several times faster than the state-of-the-art LCC approach. Unlike prior coded computing approaches, AVCC decouples the computational integrity check from the straggler tolerance thereby reducing the cost of Byzantine tolerance.  

AVCC opens the door for several interesting future directions. The encoding, decoding, and data distribution process is conducted by a trusted central server. The question to pose next is whether this central server could also be removed from our trust base. We believe that using a trusted execution environment such as an Intel SGX~\cite{costan2016intel,Tramr2019SlalomFV,niu2021asymml} equipped cloud server, one can move the vulnerable computations such as encoding and decoding to a hardware assisted secure enclave. Second, deep neural networks have non-linear computations that are difficult to decode when such computations are applied to encoded data. One potential option is to approximate such non-linearities using polynomials or rational functions ~\cite{mishra2020delphi,ali2020polynomial,soleymani2021approxifer,jahani2020berrut}. This approximation comes at the cost of accuracy loss. However, it can defend against Byzantine workers attacks. 